\documentclass{article}
\usepackage[paper=a4]{typearea}
\usepackage[utf8]{inputenc}
\usepackage[T1]{fontenc}
\usepackage{graphicx}
\usepackage{longtable}
\usepackage{wrapfig}
\usepackage{rotating}
\usepackage[normalem]{ulem}
\usepackage{amsmath}
\usepackage{amsthm}
\usepackage{amssymb}
\usepackage{capt-of}
\usepackage{hyperref}
\usepackage{physics}
\usepackage{dsfont}
\usepackage{enumerate}
\usepackage[margin=2cm]{geometry}
\usepackage{xcolor}

\date{}

\usepackage{mathtools}
\usepackage{comment}
\newcommand{\widebar}[1]{\mkern1.5mu\overline{\mkern-1.5mu#1\mkern-1.5mu}\mkern1.5mu}

\usepackage{graphicx}

\newcommand{\QED}{\hfill$\square$}

\title{Efficiently Computable Limits on EPR Pair Generation \\ in Quantum Broadcast Channels}

\usepackage{authblk}

\author[1]{Patrick Hayden}
\author[2,3]{Debbie Leung}
\author[4,5]{Hjalmar Rall}
\author[3,6]{Farzin Salek}

\affil[1]{Department of Physics, Stanford University, Stanford, CA, USA}
\affil[2]{Perimeter Institute for Theoretical Physics, Ontario, Canada}
\affil[3]{Institute for Quantum Computing (IQC), University of Waterloo, Ontario, Canada}
\affil[4]{Department of Mathematics, Technical University of Munich, Munich, Germany}
\affil[5]{Munich Center for Quantum Science and Technology (MCQST), Munich, Germany}
\affil[6]{Dahlem Center for Complex Quantum Systems, Freie Universit\"{a}t Berlin, Berlin, Germany}

\newtheorem{theorem}{Theorem}
\newtheorem{proposition}[theorem]{Proposition}
\newtheorem{lemma}[theorem]{Lemma}
\newtheorem{remark}[theorem]{Remark}
\newtheorem{corollary}[theorem]{Corollary}
\newtheorem{definition}[theorem]{Definition}

\begin{document}

\maketitle

\begin{abstract}
We investigate the generation of EPR pairs between three observers in a general causally structured setting, where communication occurs via a noisy quantum broadcast channel. The most general quantum codes for this setup take the form of tripartite quantum channels. Since the receivers are constrained by causal ordering, additional temporal relationships naturally emerge between the parties. These causal constraints enforce intrinsic no-signalling conditions on any tripartite operation, ensuring that it constitutes a physically realizable quantum code for a quantum broadcast channel. We analyze these constraints and, more broadly, characterize the most general quantum codes for communication over such channels. We examine the capabilities of codes that are fully no-signalling among the three parties, positive partial transpose (PPT)-preserving, or both, and derive simple semidefinite programs to compute the achievable entanglement fidelity. We then establish a hierarchy of semidefinite programming converse bounds—both weak and strong—for the capacity of quantum broadcast channels for EPR pair generation, in both one-shot and asymptotic regimes. Notably, in the special case of a point-to-point channel, our strong converse bound recovers and strengthens existing results. Finally, we demonstrate how the PPT-preserving codes we develop can be leveraged to construct PPT-preserving entanglement combing schemes, and vice versa.
\end{abstract}

\section{Introduction}

The most general quantum codes for generating entanglement between two parties, a form of quantum communication processing equivalent to entanglement transmission or subspace transmission \cite{Devetak2005-jn}, are modelled as bipartite quantum channels \cite{shannon1961two,Childs2006-ge}. %
This means that while the sender $A$ and receiver $B$ communicate through a noisy quantum channel, their most general encoding and decoding strategies can be described by a bipartite quantum channel. This framework accounts for possible cooperation between the two players, potentially granting them additional resources beyond the noisy channel itself.  
When the bipartite operation consists only of local operations—an encoder for the sender and a decoder for the receiver—the code is referred to as an \textbf{unassisted code (UN)}, since both players operate independently in their own laboratories and rely solely on the noisy channel for quantum communication.  
The study of unassisted quantum communication has been a fundamental topic in quantum information theory, particularly in understanding the capacity of a quantum channel to transmit information from $A$ to $B$. Among various aspects, significant attention has been given to quantum capacity—the maximum rate at which quantum information can be reliably transmitted. Assuming an asymptotically large number of channel uses, this question was answered by the LSD theorem \cite{Lloyd-S-D,L-Shor-D,L-S-Devetak}. The theorem establishes that the quantum capacity is given by the regularized coherent information, an expression that requires an optimization over infinitely many channel uses.   

In light of this result, a promising direction is to deepen our understanding of quantum capacity for channels where the LSD formula is known to be additive. According to the achievability part of the LSD theorem, if the communication rate is below the channel’s quantum capacity, there exists a coding scheme whose decoding error—measured by the fidelity—vanishes exponentially with the number of channel uses. Conversely, the theorem’s converse asserts that if the communication rate exceeds the quantum capacity, then the fidelity of any coding scheme remains strictly less than one in the asymptotic limit. This type of result is known as a \textit{weak converse}, as it allows for the possibility of a trade-off between communication rate and fidelity—suggesting that one might exceed the quantum capacity by tolerating a higher error, i.e., lower fidelity. In contrast, a \textit{strong converse} theorem rules out such trade-offs: it asserts that if the communication rate exceeds the channel’s capacity, the fidelity of any coding scheme must drop sharply, vanishing in the asymptotic limit of many channel uses. In this sense, the capacity acts as a sharp boundary—marking a clear threshold between achievable and unachievable rates—beyond which the quality of transmission deteriorates abruptly, reminiscent of a phase transition.
For classical communication, the strong converse property has been established for several classes of quantum channels \cite{Winter_1999,Wilde_2014,Wilde_2014-3}. For quantum communication, only a few cases have been studied \cite{Tomamichel_2017,Fang_2021,cheng2024strongconversetheoremsquantum,Leditzky_2016}. In particular, it was shown in \cite{Tomamichel_2017} that the Rains information \cite{Rains2001} serves as a strong converse rate for quantum information transmission.

In practical scenarios, resources are limited, and the number of times a channel can be used is finite. Additionally, the channel may not remain consistent across uses (i.e., it may not be memoryless), making it crucial to optimize communication strategies under these constraints. These challenges have driven significant interest in the non-asymptotic (finite blocklength) regime, which examines the trade-offs between the size of the transmitted quantum system, the number of channel uses, and the achievable fidelity \cite{HAYASHI_2017}.
There has been a surge of interest in understanding information-theoretic tasks beyond the asymptotic limit, leading to significant progress in both classical information theory \cite{5290292,5452208,CIT-086,6269084},  and, more recently, in quantum information theory \cite{Datta2016-xu,6574274,Wilde_2017,8861115,Matthews2014} \cite{PhysRevLett.108.200501,5967913,5961832}, where the challenges of finite blocklength analysis are particularly pronounced.

A key limitation of several bounds developed in these works is the lack of explicit or computationally efficient methods for evaluating the relevant quantities. 
 In many cases, it is unclear whether such computations are feasible at all. This underscores a broader and fundamental challenge—of both theoretical and practical importance—namely, the difficulty of obtaining quantum capacity expressions that are not only well-defined but also amenable to efficient computation.
Gaining insight into this generally intractable problem has led to a natural line of inquiry: analyzing how the inclusion of additional free resources in the coding scheme can enhance performance and potentially simplify capacity evaluation. While this topic is interesting in its own right—for instance, by considering quantum capacity when the sender and receivers have free access to entanglement—it also provides useful upper bounds on unassisted quantum capacity. This approach has been pursued in works such as \cite{Anshu_2019,Datta_2013}, though these results still face challenges related to uncomputability. 
A similar strategy has been employed in the context of entanglement distillation \cite{Rains2001}, where one-shot converse bounds under local operations and classical communication (LOCC) are derived. These bounds are expressed as semidefinite programs (SDPs) by considering the broader class of positive partial transpose (PPT)-preserving operations \cite{PhysRevA.60.179}. 

The search for computable bounds on quantum communication naturally leads to finding codes that allow for efficient optimization. A particularly effective approach is to impose restrictions on the bipartite quantum channel (the code), ensuring that the resulting restricted code enables efficiently computable upper bounds on unassisted quantum capacity.  This raises the fundamental question: what is the most general bipartite operation that can serve as a quantum code?  
It turns out that for a bipartite channel to qualify as a quantum code, it only needs to preserve causality \cite{Chiribella_2008,duan-winter}: The sender's action can directly influence the receiver’s, as it occurs earlier in time, but the reverse is not true. This causal ordering is enforced by a no-signalling condition, ensuring that the receiver cannot send information back to the sender. Apart from this, there are no further restrictions on the bipartite operation. Inspired by the pioneering work of Rains \cite{Rains2001}, this framework is referred to as \textbf{forward-assisted codes} \cite{leung-matthews}, which consist of all forward communication and local operations in the one-shot regime—where multiple channel uses may occur in parallel, without feedback or adaptive coding strategies. 

It is evident that the class of forward-assisted codes is too powerful to be meaningful, as they may allow the two parties to even bypass the noisy channel entirely and use the code itself to transmit information. Therefore, additional constraints must be imposed on the bipartite channel, leading to different classes of quantum codes. 
To facilitate the derivation of computable capacity bounds, prior studies have identified two particularly useful classes of quantum codes: \textbf{no-signaling (NS) codes} and \textbf{PPT-preserving codes}.
Specifically, NS codes encompass all assisting resources that do not enable communication between parties, including entanglement-assisted codes, while PPT-preserving codes include all LOCC operations. We briefly review these frameworks, which have been shown to admit semidefinite programming formulations for bounding code performance.  

Extending this line of work, \cite{leung-matthews} investigated the capabilities of no-signalling and PPT-preserving codes for quantum communication, expressing the entanglement fidelity (also known as channel fidelity) as an SDP to evaluate their performance—an approach that has since served as a key stepping stone for subsequent developments. The quantitative gap between the unassisted entanglement fidelity and the fidelity achievable with PPT-preserving or no-signalling codes was later analysed in \cite{Berta2022-aq}, presenting asymptotically converging hierarchy of SDP relaxations, enabling a precise quantification of the advantage offered by these assisted coding strategies. Building on the SDP formulation developed in \cite{leung-matthews}, \cite{wang-q} introduced the one-shot quantum capacity with a small error tolerance for PPT-preserving and no-signalling codes, characterizing it as an optimization problem. Using this framework, they derived SDP-based bounds to evaluate the one-shot capacity for a given error threshold. While \cite{Tomamichel2016} had previously provided efficiently computable converse bounds, the SDP bounds from \cite{wang-q} are generally tighter and can be strictly tighter for key examples such as the qubit amplitude damping channel and the qubit depolarizing channel. For efficiently computable limits on the classical capacity of quantum channels, see \cite{8012535,8638816}. Motivated by these advances, we turn our attention to the non-asymptotic analysis of quantum broadcast channels, extending this growing body of work to the more intricate setting of multi-user quantum communication. We develop this perspective in the discussion that follows.

\medskip

As quantum networks continue to expand, new communication scenarios become increasingly relevant. One fundamental setting involves a single transmitter sending information to two receivers over a shared channel with one input and two outputs. This model is known as a quantum broadcast channel \cite{q-yard}, and it represents a basic ingredient of any quantum network. The concept of broadcast channels was first introduced nearly five decades ago \cite{cover-broadcast} in the context of classical information theory, where channels were modeled by conditional probability distributions. In the classical setting, a transmitter aims to send two independent messages to two separate receivers, and the central question is to determine the set of all simultaneously achievable transmission rates. Despite extensive research, this problem remains largely unsolved in general~\cite{4797571,6145471,5673931}.

The importance of this network primitive has motivated its generalization from the classical to the quantum domain, giving rise to the classical-quantum (cq-) broadcast channel model \cite{q-yard}. A quantum broadcast channel is formally defined as a quantum channel with one input and two outputs: the sender, Alice, transmits a quantum state in Hilbert space $A$; the first receiver, Bob, receives a state in Hilbert space $B$; and the second receiver, Charlie, receives a state in Hilbert space $C$. This model captures scenarios in which the sender transmits classical information, resulting in quantum outputs for the receivers that are conditioned on the input. However, the noncommutativity of quantum output states makes characterizing the classical capacity significantly more challenging than in the classical case~\cite{salek2024threereceiverquantumbroadcastchannels,q-yard,savov-wilde,7412732}. In particular, the quantum nature of conditional structures leads to converse bounds that fail to tightly match achievable rate regions~\cite{6634255,8370123}.

When it comes to quantum information transmission over a quantum broadcast channel—where Alice’s input results in the two receivers obtaining marginals of an arbitrarily mixed bipartite quantum state—the problem becomes notoriously difficult. To date, the only known result involves transmitting the same coherent quantum message over isometric quantum channels \cite{q-yard}. This effectively amounts to generating a GHZ state shared among Alice, Bob, and Charlie, where the three parties are connected via an isometric channel. However, this setup is somewhat artificial, as Charlie is technically the inaccessible environment, reinterpreted as a legitimate receiver only because of the isometric nature of the channel.

There are two levels of difficulty in this context. The first is generating a Greenberger–Horne–Zeilinger (GHZ) state over a quantum broadcast channel, which enables the transmission of the same quantum message from Alice to both Bob and Charlie. The second, more challenging task is the simultaneous generation of two independent EPR pairs—one between Alice and Bob, and the other between Alice and Charlie—which would allow Alice to send independent quantum messages to each receiver. We expect the first task to be more tractable, while the second is significantly harder—not least because the classical capacity region is a special case of this problem.

Despite some progress, fully characterizing the quantum broadcast capacity remains an open challenge. Ongoing research continues to refine both inner and outer bounds~\cite{5466521,Bäuml_2017,7438836,9153029,Colomer2024-ro,cheng2023quantum,wang2017}.
A subset of the current authors will present new achievable rate regions for the quantum broadcast capacity in the most general setting of EPR pair generation~\cite{3-salek-hayden}. However, while our upcoming paper and existing work have significantly advanced our understanding, the current bounds remain largely intractable—either due to the need for regularization (which makes the capacity impossible to compute explicitly) or because of the inherent difficulty in evaluating the relevant entropic quantities. For example, \cite{seshadreesan2016} provides an upper bound on the quantum capacity in terms of a regularized quantity that, in general, is not computable.

Inspired by the existence of efficiently computable bounds for the capacity of point-to-point quantum channels, it is natural to seek similar bounds for quantum broadcast channels. In this more complex setting, the most general quantum code takes the form of a tripartite quantum operation, requiring an extension of the bipartite code models. This leads to the study of tripartite quantum operations as quantum codes aimed at generating EPR pairs between three parties. Unlike bipartite operations, tripartite operations involve more intricate definitions of no-signalling constraints. Depending on which specific constraints are imposed, an operation may qualify as a valid quantum code for a given task while still respecting causality. As we will show, this framework yields the most general class of quantum codes for generating EPR pairs between one sender and two receivers connected via a quantum broadcast channel. The structural constraints that define bipartite quantum codes can be naturally generalized to the tripartite case, giving rise to multiple classes of codes that are amenable to efficient computation. For the task of classical information transmission over quantum broadcast channels, this approach was explored in \cite{xie2017converse}, where the authors derived semidefinite programming bounds on the classical capacity using no-signalling and/or PPT-preserving codes.

The transmission of quantum information over a quantum channel is closely related to the task of entanglement distillation—a foundational area of quantum information theory focused on transforming multipartite quantum states into more useful or better-understood forms of entanglement. The objective is often to convert a given quantum state into a standard entangled state such as an EPR pair or a GHZ state, either to enhance utility or to gain insight into its entanglement structure. Depending on the initial and target states, available local operations, and the number of involved parties, various distillation protocols have been developed. 

One particularly significant form of entanglement distillation is entanglement combing. In this protocol, a multipartite entangled state is systematically restructured into a collection of EPR pairs, where a central "root" party shares bipartite entanglement with multiple "leaf" parties \cite{PhysRevLett.103.220501}. This is accomplished through sequential applications of state merging \cite{state-merging}, progressively merging the leaf parties into the root. The protocol also generalizes to mixed states \cite{salek2023newprotocolsconferencekey,dutil2011assistedentanglementdistillation}. In the asymptotic setting for pure states, \cite{PhysRevLett.103.220501} characterized the achievable EPR pair rate region as the convex hull of its extreme points, each corresponding to a specific order in which merging is performed. Traditionally, interior points of this region were accessed through time-sharing between these extreme protocols. However, recent insights from \cite[Sec.~6.3 \& Cor.~6.12]{Colomer2024-ro} show that simultaneous merging steps can directly achieve the entire rate region, eliminating the need for time-sharing. Advances for mixed-state combing protocols are further developed in \cite{salek2023newprotocolsconferencekey}. 

Our interest in entanglement combing stems from its conceptual and operational parallels with the generation of EPR pairs over quantum broadcast channels. In both scenarios, the goal is for a sender—Alice—to establish entanglement with two receivers, Bob and Charlie. In entanglement combing, the starting resource is a shared quantum state among the three parties, while in the broadcast setting, it is a noisy quantum channel with one input and two outputs. For point-to-point channels, a key connection between EPR distillation and quantum coding was established in \cite{PhysRevA.60.179, Rains2001}, where it was shown that PPT-preserving codes correspond to PPT-preserving distillation protocols via the Choi state of the channel. Building on this idea, \cite{leung-matthews} demonstrated that any PPT-preserving distillation protocol acting on the Choi state can equivalently be interpreted as a PPT-preserving quantum code for the channel. As a consequence, the maximum fidelity achievable by PPT-preserving codes upper bounds that of PPT-preserving distillation protocols. Furthermore, they provided criteria under which the converse holds—that is, when a PPT-preserving code can be realized via a corresponding distillation protocol, thereby achieving the same fidelity and saturating the bound.

\medskip

\textbf{Summary of our contribution:}
We study tripartite quantum channels as quantum codes in quantum networks, focusing on their no-signalling properties and PPT-preserving conditions. By analysing these constraints, we identify key structural properties of quantum broadcast coding.  We establish the most general code for quantum broadcast channels, clarifying aspects that remained implicit in prior work on point-to-point channels. A central result of our work is the formulation of channel fidelity for quantum broadcast channels assisted by PPT-preserving and/or NS codes as a semidefinite program. Using this SDP framework, we derive efficiently computable weak converse bounds on the quantum capacity of quantum broadcast channels for both the asymptotic and non-asymptotic settings. Our approach follows a layered strategy: We begin by defining a broad class of forward-assisted tripartite quantum codes, then restrict optimization to more structured NS and/or PPT-preserving assisted capacities, and finally relax the problem into an efficiently solvable SDP. Unlike in the classical case, quantum capacity optimization presents intricate constraints that do not straightforwardly reduce to an SDP. Addressing this challenge, we analyse multipartite NS conditions in terms of their role in enforcing causality, ensuring that our relaxations remain both meaningful and physically consistent, i.e., no relaxed NS-assisted code should violate causality.

We further introduce the one-shot quantum sum-capacity under PPT-preserving and NS-assisted codes, characterizing it as a nonlinear optimization problem. From this formulation, we derive SDP upper bounds on the one-shot capacity for a given error tolerance. For a relaxed variant of the $NS \cap PPT$ assisted capacity, we prove its additivity under tensor products, leading to a strong converse theorem: If the sum-rate exceeds capacity, fidelity vanishes in the asymptotic limit. This result generalizes and strengthens prior work on point-to-point channels \cite{wang-q}. Drawing inspiration from techniques developed for finite-resource bipartite quantum capacity \cite{wang-q} and classical message transmission over quantum broadcast channels \cite{xie2017converse}, we extend the SDP framework to account for the additional complexity of sending quantum information over quantum broadcast channels. Finally, building on prior work in entanglement distillation \cite{leung-matthews}, we explore how PPT-preserving codes can be used to simultaneously distill EPR pairs between Alice and both Bob and Charlie, given an initial shared quantum state. More precisely, we establish a connection between PPT-preserving coding and PPT-preserving entanglement combing, demonstrating when and how a PPT-preserving code can be leveraged as an entanglement combing protocol and vice versa.

\medskip

The paper is organized as follows: Section \ref{preliminaries} introduces basic concepts and notation. In Section \ref{class-of-codes}, we classify different types of codes and examine the interplay between causality preservation and no-signalling conditions in multipartite settings. Section \ref{sdp-fidelity} formulates semidefinite programs to evaluate the fidelity of PPT-preserving and no-signalling codes. In Section \ref{weak-converse}, we derive an optimization problem for the one-shot capacity and introduce a hierarchy of SDP relaxations to establish weak converse bounds, while Section \ref{sec:strong-converse} presents our strong converse bound. The connection between PPT-preserving entanglement combing protocols and PPT-preserving codes is explored in Section \ref{combing}. Finally, we conclude with discussions and future directions in Section \ref{discussion}, while additional material is provided in the Appendix \ref{appendix}.

\section{Preliminaries}
\label{preliminaries}
We begin by recalling some basic notions and notations before presenting more advanced background. This approach aligns with standard treatments in the literature, such as those found in the excellent textbooks \cite{khatri-wilde-book,Tomamichel_2016,Wilde_2013,gour2024resourcesquantumworld,Hayashi2017-cv}
We denote the Hilbert spaces associated with finite-dimensional quantum systems by capital letters, such as $A,B$, etc., and represent the dimension of a Hilbert space $A$ by $\abs{A}$. We will write the systems on which an operator is defined as a superscript. The same convention applies to all operators, except trace $\Tr $ and transpose operators $T$, for which we will use subscripts. An operator $\rho^{A}$ is a density operator, representing the state of a quantum system, if it is positive semidefinite, denoted $\rho^{A}\geq 0$, and $\Tr\rho^{A}=1$. The identity operator on Hilbert space $A$ is denoted by $\mathds{1}^{A}$. For operators $X$ and $Y$, we use the notation $\abs{X} \leq Y$ to denote that $-Y \leq X \leq Y$.
 The composition of two systems is represented by the tensor product of their Hilbert spaces, $AB = A \otimes B$. For example,
for a multipartite operator $\rho^{AB}$ acting on tensor product space $A\otimes B$, we may write $\rho^{AB}\in A\otimes B$. 
 For an operator $\rho^{AB}$, the corresponding reduced operator on $A$ is denoted as $\rho^A = \Tr_B \rho^{AB}$.  
For an operator $X^{A}$, its trace norm is defined as $\|X^{A}\|_1 = \Tr \sqrt{XX^\dagger}$. The fidelity between two state $\rho$ and $\sigma$ is defined as $F(\rho, \sigma) = \|\sqrt{\rho}\sqrt{\sigma}\|_1 = \Tr\sqrt{\rho\sigma}$. Note that when $\psi = \ketbra{\psi}$ is pure, the fidelity simplifies to $F(\psi,\sigma) = \bra{\psi}\sigma\ket{\psi}$. 
For two Hilbert spaces of the same dimension, we use the same letter with a tilde to distinguish them. For instance, $A$ and $\widetilde{A}$ represent two isometric Hilbert spaces. Additionally, operators with a prime symbol are intended to belong to the same local party, though their dimensions may differ. For example, $A$ and $A'$ could denote Alice's input and output spaces, both belonging to her locally.

A bipartite state $\rho^{AB}$ is called entangled if it is not a separabel state, i.e. if it cannot be written in the form
\begin{align*}
    \rho^{AB} = \sum_i \alpha_i \, \sigma_i^A \otimes \omega_i^B,
\end{align*}
where $\sigma_i^A$ and $\omega_i^B$ are states on $A$ and $B$, respectively, and $\alpha_i$ are probabilities. In other words, a state on $AB$ is entangled if it cannot be expressed as a probabilistic mixture of product states on $A$ and $B$.  
Separable states form a subclass of the broader family of positive partial transpose (PPT) states. A bipartite quantum state is said to be PPT if it remains positive semi-definite under partial transposition with respect to one of its subsystems. The subset of PPT states that are not separable are known as bound entangled states \cite{PhysRevLett.80.5239}. These are entangled quantum states from which no maximally entangled pairs (EPR pairs) can be distilled using two-way LOCC operations.

An EPR pair, or maximally entangled state, is given by $\phi^{A\widetilde{A}} = \ketbra{\phi}^{A\widetilde{A}}$, where
\begin{align*}
    \ket{\phi}^{A\widetilde{A}} = \frac{1}{\sqrt{|A|}} \sum_i \ket{i}^A \otimes \ket{i}^{\widetilde{A}}.
\end{align*}
Here, $|A| = |\widetilde{A}|$, and $\{\ket{i}^A\}$ and $\{\ket{i}^{\widetilde{A}}\}$ are standard orthonormal bases for $A$ and $\widetilde{A}$, respectively. EPR pairs do not belong to the PPT class and are a central resource in quantum information processing.
For convenience, we also define the unnormalized EPR pair as
\begin{align*}
    \ket{\widebar{\phi}}^{A\widetilde{A}} = \sum_i \ket{i}^A \otimes \ket{i}^{\widetilde{A}}.
\end{align*}
This unnormalized state will prove useful in simplifying expressions, such as in the definition of the Choi matrix.

For any state $\sigma^{A\widetilde{A}}$, its entanglement fidelity is defined as $F(\phi^{A\widetilde{A}},\sigma^{A\widetilde{A}})=\Tr\sigma^{A\widetilde{A}}\phi^{A\widetilde{A}}$. The partial transpose of an EPR pair results in an operator that includes negative eigenvalues. Specifically:  
\begin{align*}
    (\phi^{A\widetilde{A}})^{T_{A}} = \frac{1}{\abs{A}}(\mathbb{S}^{A\widetilde{A}} - \mathbb{A}^{A\widetilde{A}}),
\end{align*}
where $\mathbb{S}^{A\widetilde{A}}$ and $\mathbb{A}^{A\widetilde{A}}$ are the projectors onto the symmetric and anti-symmetric subspaces of $A\otimes \widetilde{A}$, respectively. These projectors satisfy the relation $\mathbb{S}^{A\widetilde{A}} + \mathbb{A}^{A\widetilde{A}} = \mathds{1}^{A\widetilde{A}}$, a property specific to bipartite operators. In terms of the flip operator $F^{A\widetilde{A}} = \sum_{i,j}\ket{i,j}\bra{j,i}^{A\widetilde{A}}$, the projectors can be expressed as $\mathbb{S}^{A\widetilde{A}} = \frac{1}{2}(\mathds{1}^{A\widetilde{A}} + F^{A\widetilde{A}})$ and $\mathbb{A}^{A\widetilde{A}} = \frac{1}{2}(\mathds{1}^{A\widetilde{A}} - F^{A\widetilde{A}})$.

Another useful property of the EPR pair, often referred to as the transpose trick, is the following identity. For any operator $X^{A}$, we have
\begin{align}
\label{eq:transposetrick}
    \big((X^{A})^{T} \otimes \mathds{1}^{\widetilde{A}}\big) \ket{\phi}^{A\widetilde{A}} = 
    \big(\mathds{1}^{A}\otimes X^{\widetilde{A}} \big)\ket{\phi}^{A\widetilde{A}}.
\end{align}
The EPR pair can also be used to express the trace of an operator as follows:
\begin{align}
\label{trace-EPR}
    \Tr X^A = \frac{1}{\abs{A}} \Tr(X^A \otimes \mathds{1}^{\widetilde{A}}) =  \abs{A}\bra{\phi}^{A\widetilde{A}} (X^A \otimes \mathds{1}^{\widetilde{A}}) \ket{\phi}^{A\widetilde{A}}.
\end{align}
Note that the identity for the trace of an operator in Eq. \eqref{trace-EPR} holds for any member of the Bell state family. For example, in the qubit case, one can replace $\ket{\phi}^{A\widetilde{A}}$ in Eq. \eqref{trace-EPR} with either $\frac{1}{\sqrt{2}}(\ket{00}-\ket{11})$, $\frac{1}{\sqrt{2}}(\ket{01}+\ket{10})$, or $\frac{1}{\sqrt{2}}(\ket{10}-\ket{01})$. However, and importantly, the transpose trick in Eq. \eqref{eq:transposetrick} only applies to the EPR pair $\ket{\phi} = \frac{1}{\sqrt{2}}(\ket{00}+\ket{11})$. This has important consequences for the study of teleportation, as we will see later in Sec. \ref{combing}.    

\medskip

The cyclicity property of the (partial) transpose operation within the (partial) trace states that for any operators $L^{AB}$ and $H^{AB}$, we have \begin{align*} \Tr_A\left\{(L^{AB})^{T_A} H^{AB} \right\} = \Tr_A\left\{ L^{AB} (H^{AB})^{T_A} \right\}, \end{align*} but this is not necessarily equal to $\Tr_A\left\{ L^{AB} H^{AB} \right\}$. Notably, the transpose operation cannot generally be removed even when applied to the entire system. Specifically, while we have
\begin{align*} 
 \Tr\left\{(L^{AB})^{T_{AB}} H^{AB} \right\} = \Tr\left\{ L^{AB} (H^{AB})^{T_{AB}} \right\}, \end{align*} 
this is not necessarily equal to $\Tr\left\{ L^{AB} H^{AB} \right\}$, even if $L^{AB}$ and $H^{AB}$ are Hermitian. As a simple counterexample, let $Y = \sigma_y$ be the Pauli $Y$ operator, and define $L = (I + Y)/2$ and $H = (I - Y)/2$. Clearly, we have $L^T = H$, yet $\Tr L^T H = 1$, whereas $\Tr L H = 0$.

\medskip
A quantum channel $\mathcal{N}^{A'\to B}$ is a linear, completely positive (CP) and trace-preserving (TP) map which accepts input quantum states in $A'$ and produces output quantum states in $B$. Quantum channels can be represented in various forms, including the Stinespring representation, Kraus operators, and the Choi-Jamio\l{}kowski (commonly abbreviated as Choi) representation. In this paper, we extensively use the Choi representation. The Choi representation of a channel $\mathcal{N}^{\widetilde{A',} \to B}$ is the operator $N^{A'B} \in A' \otimes B$, defined as: 

\begin{align*} 
N^{A'B} &= |A'| \, (\text{id}^{A'} \otimes \mathcal{N}^{\widetilde{A'} \to B})(\phi^{A'\widetilde{A'}})\\
& = (\text{id}^{A'} \otimes \mathcal{N}^{\widetilde{A'} \to B})(\widebar{\phi}^{A'\widetilde{A'}})
\end{align*} 
where $\text{id}^{A'}$ is the identity channel on $A'$, and $\ket{\widebar{\phi}}^{A'\widetilde{A'}} = \sum_{i}\ket{i,i}^{A'\widetilde{A'}}$ is the unnormalized maximally entangled state. (In the literature, $N^{A'B}$ is sometimes referred to as Choi matrix with respect to unnormalized EPR pair). Note that $N^{A'B}$ is not a quantum state, as it results from applying a quantum channel to an unnormalized state. We refer to this operator as the Choi matrix of the channel. In some cases, it is convenient to define the Choi representation using the normalized Bell state. For example, in Sec.~\ref{combing}, we denote the corresponding Choi state with a tilde, $\widetilde{N}^{A'B}$, which relates to the Choi matrix as  
\begin{align*}
   \widetilde{N}^{A'B} = \frac{N^{A'B}}{|A'|}. 
\end{align*}
Throughout this paper, we adopt the convention that quantum channels are denoted by calligraphic letters (e.g., $\mathcal{N}^{A'\to B}$), while their corresponding Choi matrices are represented by the same letter in regular font (e.g., $N^{A'B}$). Additionally, note that we may use the labels $A'$ and $\widetilde{A'}$ interchangeably, as they correspond to isometric spaces. 
It is a well-established result that the channel $\mathcal{N}^{A'\to B}$ is completely positive if and only if $N^{A'B} \geq 0$, and it is trace-preserving if and only if $\Tr_B(N^{A'B}) = \mathds{1}^{A'}$. By employing the Choi matrix, the action of the channel $\mathcal{N}^{A'\to B}$ can be expressed as
\begin{align*}
    \mathcal{N}^{A'\to B}(\rho^{A'}) = \Tr_{A'}\left\{N^{A'B} \big((\rho^{A'})^{T} \otimes \mathds{1}^B\big)\right\},
\end{align*}
where the transpose is taken with respect to the basis chosen for the maximally entangled state $\ket{\phi}^{A'\widetilde{A'}}$. An important property of a quantum channel $\mathcal{N}^{A' \to B}$ is its \textit{channel fidelity}. This quantity is defined as the fidelity between an EPR pair $\phi^{A'\widetilde{A}'}$ and the state obtained by sending one half of the EPR pair through the channel. In essence, channel fidelity quantifies how well a channel preserves the entanglement of an EPR pair. This definition assumes that the channel's output space is isomorphic to its input space. 
Using the definitions of the Choi matrix and pure state fidelity, the channel fidelity of $\mathcal{N}^{A' \to B}$ is given by  
$\frac{1}{|A'|} \Tr \big( N^{A'B} \phi^{A'\widetilde{A}'} \big)$.

A quantum broadcast channel $N^{A' \to BC}$ refers to a quantum channel with one input and two outputs that are physically separated. For simplicity, we often think of the users as people, with Alice controlling the input $A'$, while Bob $B$ and Charlie $C$ are associated with the outputs. The Choi matrix of the broadcast channel $N^{A' \to BC}$ is defined as before by applying the channel to half of an EPR pair $\phi^{A'\widetilde{A'}}$. When Alice sends the state $\rho^{A'}$ through the channel, using the Choi representation of the channel, Bob and Charlie receive the following states, respectively:
\begin{align*}
   \Tr_C \big\{\mathcal{N}^{A'\to BC}(\rho^{A'})\big\} &= \Tr_{A'C}\left\{N^{A'BC} \big((\rho^{A'})^T \otimes \mathds{1}^{BC}\big)\right\},\\
    \Tr_B \big\{\mathcal{N}^{A'\to BC}(\rho^{A'})\big\} &= \Tr_{A'B}\left\{N^{A'BC} \big((\rho^{A'})^T \otimes \mathds{1}^{BC}\big)\right\}.
\end{align*}
Moving on, a tripartite channel $\mathcal{Z}^{ABC \to A'B'C'}$ is a three-input three-output channel where Alice controls both the input system $A$ and the output system $A'$, Bob controls both the input system $B$ and the output system $B'$, and Charlie controls both the input system $C$ and the output system $C'$. The Choi matrix of this channel is defined by preparing three EPR pairs of the appropriate sizes and sending half of them through the channel, as needed (note that we additionally multiply by the dimensions of the EPR pairs, so in this paper, the Choi matrix is actually defined with respect to unnormalized EPR pairs):
\begin{align}
\label{tri-Choi}
    Z^{ABCA'B'C'} = \abs{A}\abs{B}\abs{C}\left(\text{id}^{A}\text{id}^{B}\text{id}^{C}\otimes 
\mathcal{Z}^{\widetilde{A}\widetilde{B}\widetilde{C}\to A'B'C'}\right) (\phi^{A\widetilde{A}} \otimes \phi^{B\widetilde{B}} \otimes \phi^{C\widetilde{C}}).
\end{align}

\bigskip
Choi states (or matrices) reveal important properties of quantum channels when analyzed as quantum states through the lens of entanglement theory. A channel $\mathcal{N}^{A' \to B}$ is called binding entanglement if its Choi state is bound entangled. Equivalently, the necessary and sufficient conditions for a channel to be bound entangled are: (i) its two-way LOCC-assisted quantum capacity is zero, and (ii) it is possible to generate bipartite bound entangled states using one or more instances of the channel along with LOCC operations \cite{Horodecki2000-pi}. Another important class of channels, which includes binding entanglement channels, is the class of PPT channels. A channel $\mathcal{N}^{A' \to B}$ is referred to as a PPT channel if its Choi state is a PPT state, meaning $(\widetilde{N}^{A'B})^{T_{A'}} \geq 0$. Equivalently, such point-to-point channels always output PPT states for any bipartite input: 
\begin{align*} 
\left(\text{id}^{R}\otimes \mathcal{N}^{A' \to B} (\rho^{RA'}) \right)^{T_B} \geq 0, \end{align*} 
and therefore preserve the set of PPT states \cite{PhysRevA.60.179}. While some references use the term PPT-preserving channels to describe such maps \cite{Tomamichel2016}, following, e.g., \cite{Rains2001,leung-matthews}, we reserve that term for multipartite scenarios and refer to the point-to-point case simply as PPT channels.

For channels with multiple inputs and outputs, the notion of being PPT (positive partial transpose) is extended by considering the action of partial transposition across any bipartite cut. In the bipartite setting, these ideas were studied in \cite{Rains2001,PhysRevA.60.179}, where the term PPT-preserving channels was introduced. We now describe this concept through the example of a tripartite channel.
Specifically, a tripartite channel $\mathcal{Z}^{ABC \to A'B'C'}$ is said to be PPT-preserving with respect to the $A$-$BC$ partition if it maps any state that is PPT across that cut to another state that is also PPT across the same cut. We denote this property by PPT$_A$.
Analogously, one can define PPT-preserving properties with respect to the $B$-$AC$ and $C$-$AB$ partitions, denoted by PPT$_B$ and PPT$_C$, respectively. As in the bipartite case \cite{Rains2001}, each of these properties can be characterized by the Choi matrix of the channel. In particular, a tripartite channel is PPT-preserving with respect to the $A$-$BC$ cut (i.e., it satisfies PPT$A$) if and only if its Choi matrix $Z^{ABCA'B'C'}$ is PPT across the same cut:
\begin{align*}
\left(Z^{ABCA'B'C'}\right)^{T{AA'}} \geq 0.
\end{align*}
The same condition applies to the $B$-$AC$ and $C$-$AB$ cuts, where the partial transpose is taken over $BB'$ and $CC'$, respectively. The equivalence between this Choi-based condition and the preservation of PPT states under the channel action can be established through an argument analogous to the bipartite case \cite{Rains2001}; see also \cite[Proposition 3.43]{khatri-wilde-book}. For further discussion, including an alternative definition of PPT-preserving operations, see \cite{985948}, particularly the remarks around Footnote 7.

\begin{remark}
    The condition that the Choi state (or matrix) is PPT across a given cut is stronger than requiring that the channel merely preserves the set of PPT states with respect to that cut. Namely, if the Choi matrix of a tripartite channel satisfies $(Z^{ABCA'B'C'})^{T_{AA'}} \geq 0$, one can show that
    \begin{align*}
        \left(\text{id}^{R} \otimes \mathcal{Z}^{ABC \to A'B'C'} (\rho^{RABC})\right)^{T_{A'}} \geq 0,
    \end{align*}
    for any input state $\rho^{RABC}$ with an arbitrary reference system $R$, provided that $(\rho^{RABC})^{T_A} \geq 0$. For this reason, channels whose Choi matrices are PPT with respect to a given cut are referred to in recent works as \textit{completely PPT-preserving} channels \cite{10005080,PhysRevA.103.062422}, while those that only preserve the set of PPT states are called \textit{PPT-preserving} channels; see, e.g., \cite[Sec.~13.9]{gour2024resourcesquantumworld} and \cite[Sec.~3.2]{khatri-wilde-book}. This subtle distinction is discussed in detail in \cite{Chitambar2020-ty}, where completely PPT-preserving channels are referred to as \textit{PPT maps} (see the beginning of Sec.~3.3 therein). Hence, what earlier works such as \cite{Rains2001,leung-matthews} referred to as PPT-preserving channels are now more precisely called completely PPT-preserving. Although this newer terminology is more precise, in this paper we follow the convention of \cite{Rains2001,leung-matthews} and use the term PPT-preserving to refer to what are now often called completely PPT-preserving operations.
\end{remark}

The no-signaling conditions for tripartite operations are more intricate than those for bipartite operations \cite{PhysRevA.64.052309, TEggeling_2002}. Following \cite{xie2017converse}, we define two distinct sets of no-signalling conditions, whose application to the tripartite channel can be derived by extending the definitions from bipartite case \cite{Chiribella_2008}. We use the crossed arrow notation $\not\to$ to indicate that the party(ies) on the left-hand side cannot signal to the party(ies) on the right-hand side. 
The first set comprises three two-to-one no-signalling conditions: (1) Bob and Charlie cannot signal to Alice, (2) Alice and Charlie cannot signal to Bob, and (3) Alice and Bob cannot signal to Charlie. These are denoted as $BC \not\to A$, $AC \not\to B$, and $AB \not\to C$, respectively.
The second set involves three one-to-two no-signalling conditions: (1) Alice cannot signal to Bob and Charlie, (2) Bob cannot signal to Alice and Charlie, and (3) Charlie cannot signal to Alice and Bob. These are represented as $A \not\to BC$, $B \not\to AC$, and $C \not\to AB$, respectively.
Note, importantly, that when two parties are grouped on either side of the no-signalling arrow $\not\to$, they are treated as a single, unified party. For further details including no-signalling conditions between two individual parties, refer to Section \ref{discussion}.

To explain each of the six no-signalling conditions, we analyse one condition from each set, and the analysis applies similarly to the remaining two conditions in each set. These no-signalling conditions may be expressed in equivalent ways \cite{leung-matthews,PhysRevA.74.012305,duan-winter}. We say that a tripartite operation $\mathcal{Z}^{ABC\to A'B'C'}$ is no-signalling from Bob and Charlie to Alice, i.e. $BC\not\to A$, if Alice's output cannot be influenced in any way by Bob and Charlie's inputs. This condition can be formally written as 
\begin{align}
\label{no-signalling-def-1}
    \Tr_{B'C'} \circ \mathcal{Z}^{ABC\to A'B'C'} (\cdot) = \mathcal{Z}_{a}^{A\to A'} \circ \Tr_{BC} (\cdot),
\end{align}   
for some local channel $\mathcal{Z}_{a}^{A\to A'}$ of Alice. That is, the marginal state corresponding to Alice's output is achieved by applying some local operation of Alice to her marginal share of the tripartite input state. To translate this condition into a condition on the Choi matrix of the tripartite channel, Eq. \eqref{tri-Choi}, we apply the above relation using unnormalized EPR pairs as inputs, i.e., $\abs{A}\abs{B}\abs{C}\phi^{A\widetilde{A}}\phi^{B\widetilde{B}}\phi^{C\widetilde{C}}$. Consequently, we obtain:
\begin{align}
\label{no-signalling}
    \Tr_{B'C'} Z^{ABCA'B'C'} 
    =\abs{A}\mathcal{Z}_{a}^{\widetilde{A}\to A'}(\phi^{A\widetilde{A}}) \otimes \mathds{1}^{BC},
\end{align}
where $\abs{A}\mathcal{Z}_{a}^{\widetilde{A}\to A'}(\phi^{A\widetilde{A}})$ represents the Choi matrix of Alice's local operation (defined with respect to an unnormalized EPR pair, due to the presence of $\abs{A}$, as per the convention in this paper). By further tracing out systems $B$ and $C$ from both side, we obtain the following expression for the Choi matrix of Alice's local channel in terms of $Z^{ABCA'B'C'}$:
\begin{align}
\nonumber
    Z^{AA'}_a &= \abs{A}\mathcal{Z}_{a}^{\widetilde{A}\to A'}(\phi^{A\widetilde{A}})\\ 
    \label{choi-matrix-state}
    &= \frac{1}{\abs{B}\abs{C}}\Tr_{BCB'C'} Z^{ABCA'B'C'}.
\end{align}
Substituting this expression into Eq. \eqref{no-signalling}, we derive the no-signalling condition $BC \not\to A$ as a condition on the Choi matrix as follows:
\begin{align}
\label{def-no-sig-bc-a}
    \Tr_{B'C'} Z^{ABCA'B'C'} 
     = \Tr_{BCB'C'} Z^{ABCA'B'C'} \otimes \frac{\mathds{1}^{BC}}{\abs{B}\abs{C}}.
\end{align}
We note that the coefficient $ 1/{\abs{B} \abs{C}} $ in Eq. \eqref{choi-matrix-state} arises from our choice to work with the Choi matrix rather than the (normalized) Choi state. In other words, the operator $ \Tr_{BCB'C'} Z^{ABCA'B'C'} $ is neither a normalized state nor is the Choi matrix $ Z^{ABCA'B'C'} $ itself. If instead we worked with the normalized Choi state $ \widetilde{Z}^{ABCA'B'C'} $, we would obtain $ \widetilde{Z}^{AA'}_a = \widetilde{Z}^{AA'} $. However, this distinction is ultimately irrelevant for the final expression of the no-signalling condition in Eq. \eqref{def-no-sig-bc-a}. That is, one can replace the Choi matrix $Z^{ABCA'B'C'}$ with the Choi state $\widetilde{Z}^{ABCA'B'C'}$ in Eq. \eqref{def-no-sig-bc-a}.

Similarly, we can derive the no-signalling conditions $AB\not\to C$ and $AC\not\to B$ on the Choi matrix of the tripartite channel, respectively, finding:
\begin{align*}
     \Tr_{A'B'} Z^{ABCA'B'C'} 
     = \Tr_{ABA'B'} \{Z^{ABCA'B'C'}\} \otimes \frac{\mathds{1}^{AB}}{\abs{A}\abs{B}}, \\
       \Tr_{A'C'} Z^{ABCA'B'C'} 
     = \Tr_{ACA'C'} \{Z^{ABCA'B'C'}\} \otimes \frac{\mathds{1}^{AC}}{\abs{A}\abs{C}}.
\end{align*}

The argument proceeds similarly for the second set of the no-signalling conditions. We say that a tripartite operation $\mathcal{Z}^{ABC\to A'B'C'}$ is no-signalling from Alice to Bob and Charlie, i.e. $A\not\to BC$, if 
\begin{align}
\label{no-signalling-def-2}
    \Tr_{A'} \circ \mathcal{Z}^{ABC\to A'B'C'} (\cdot) = \mathcal{Z}_{bc}^{BC\to B'C'} \circ \Tr_{A} (\cdot),
\end{align}   
for some joint bipartite channel $\mathcal{Z}_{bc}^{BC\to B'C'}$ of Bob and Charlie. 
That is, the marginal bipartite state corresponding to Bob and Charlie's outputs is achieved by applying some joint operation of Bob and Charlie to their marginal share of the tripartite input state. The Choi matrix of the bipartite channel $\mathcal{Z}_{bc}^{BC\to B'C'}$ is thus derived as $    Z_{bc}^{BC B'C'} = (1/\abs{A}) \Tr_{AA'} Z^{ABC A'B'C'}$. Using this, the no-signalling condition $A\not\to BC$ can be expressed as follows as a condition on the Choi matrix:
\begin{align*}
    \Tr_{A'} Z^{ABCA'B'C'} 
     = \Tr_{AA'} \{Z^{ABCA'B'C'}\} \otimes \frac{\mathds{1}^{A}}{\abs{A}}.
\end{align*}
Similarly for $B\not\to AC$ and $C\not\to AB$, respectively, we find:
\begin{align*}
       \Tr_{B'} Z^{ABCA'B'C'} 
    & = \Tr_{BB'} \{Z^{ABCA'B'C'}\} \otimes \frac{\mathds{1}^{B}}{\abs{B}},\\
     \Tr_{C'} Z^{ABCA'B'C'} 
    & = \Tr_{CC'} \{Z^{ABCA'B'C'}\} \otimes \frac{\mathds{1}^{C}}{\abs{C}}.
\end{align*}
We note that the final equations for the no-signalling conditions, derived as constraints on the Choi matrix, are independent of the specific definition of the Choi matrix; they remain valid even if we work with the Choi state defined with respect to (normalized) EPR pairs. In other words, in the final expressions, one can replace any Choi matrix $Z^{ABCA'B'C'}$ with the corresponding Choi state $\widetilde{Z}^{ABCA'B'C'}$. However, there were certain subtleties in deriving the final conditions, which depend on whether we use the Choi matrix or Choi state, and these have already been elaborated upon.
 Moreover, while previous work \cite{xie2017converse} may suggest that all six constraints are potentially required to establish full no-signaling between the three parties, we find that only the second set is necessary, as the conditions of the first set follow easily. One way to demonstrate this might be by generalizing the classical no-signaling condition \cite[Lemma 2.7]{hanggi} to the quantum setting. However, we provide a simpler argument instead.
\begin{lemma}
\label{no-signalling-enough}
  The no-signaling conditions ${A\not\to BC, B\not\to AC, C\not\to AB}$ establish full no-signalling between the three parties; in particular, they imply the conditions ${BC\not\to A, AB\not\to C, AC\not\to B}$
\end{lemma}
\begin{proof}
    We establish this by showing that any two conditions in the set $\{A\not\to BC, B\not\to AC, C\not\to AB\}$ imply one of the conditions in the other set. For example, if $B\not\to AC$ and $C\not\to AB$, then $BC\not\to A$.
    We begin by tracing out $C'$ from both sides of the condition corresponding to $B\not\to AC$:
    \begin{align*}
       \Tr_{C'} \Big(\Tr_{B'}\{Z^{ABCA'B'C'}\}\Big) &= \Tr_{C'} \Big(\frac{\mathds{1}^{B}}{\abs{B}} \Tr_{BB'}\{Z^{ABCA'B'C'}\}\Big)\\
       & =  \frac{\mathds{1}^{B}}{\abs{B}} \Tr_{BB'}\left\{\Tr_{C'} (Z^{ABCA'B'C'})\right\}\\
       & = \frac{\mathds{1}^{B}}{\abs{B}} \Tr_{BB'}\left\{ \frac{\mathds{1}^{C}}{\abs{C}} \Tr_{CC'}\{Z^{ABCA'B'C'}\} \right\} \\
       & = \frac{\mathds{1}^{B}}{\abs{B}} \frac{\mathds{1}^{C}}{\abs{C}} \Tr_{BCB'C'}\left\{ Z^{ABCA'B'C'} \right\},
    \end{align*}
where in the third line, we use $C\not\to AB$. Similarly, one can show that $B\not\to AC$ and $A\not\to BC$ imply $AB\not\to C$, and that $A\not\to BC$ and $C\not\to AB$ imply $AC\not\to B$.
\end{proof}

We next summarize some key concepts about representations of finite groups. Further details can be found in \cite{fulton1991representation,hayashi2017group,hayashi2017groupQI}, while a more accessible review is available in, e.g., \cite{Onorati2019}.
Let $G$ be a finite group and $V$ a finite-dimensional complex vector space. The group of invertible linear transformations on $V$ is denoted by $\text{GL}(V)$. A representation $\rho$ of $G$ on $V$ is a map  
\begin{align*}
    \rho: G \to \text{GL}(V), \quad g \mapsto \rho(g),
\end{align*} 
satisfying the property
\begin{align*}
  \rho(g)\rho(h) = \rho(gh), \quad \forall g, h \in G.  
\end{align*} 
We typically assume that $\rho(g)$ is unitary for all $g \in G$, and denote it by $U_g$.  
If there exists a non-trivial subspace $W \subseteq V$ such that
\begin{align*}
   \rho(g)w \in W, \quad \forall w \in W, \, \forall g \in G, 
\end{align*}  
then the representation $\rho$ is called reducible. The restriction of $\rho$ to $W$ is referred to as a subrepresentation. A representation is termed irreducible if no such non-trivial subspace $W$ exists.  

   Two representations $\rho$ and $\rho'$ of $G$ on vector spaces $V$ and $V'$ are equivalent if there exists an invertible linear map $T: V \to V'$ such that  
   \begin{align*}
     T \circ \rho(g) = \rho'(g) \circ T, \quad \forall g \in G.  
   \end{align*}
This equivalence is denoted as $\rho \cong \rho'$.  
For any linear map $X: V \to V$, we can define the twirl of $X$ with respect to $\rho$ as
\begin{align*}
   \widebar{X} = \frac{1}{|G|} \sum_{g \in G} U_g X U_g^\dagger. 
\end{align*}
When the representation $\rho$ decomposes into irreducible subrepresentations without multiplicities, i.e., 
\begin{align*}
    \rho(g) \cong \bigoplus_i \rho_i(g), \quad \forall g \in G,
\end{align*}  
where $\rho_i$ are irreducible and inequivalent, the twirl simplifies significantly. For any $X: V \to V$, we have
\begin{align}
\label{projection}
   \widebar{X} = \sum_i \frac{\text{Tr}(X \Pi_i)}{\text{Tr}(\Pi_i)} \Pi_i, 
\end{align} 
where $\Pi_i$ is the projector onto the subspace corresponding to the irreducible subrepresentation $\rho_i$.

\bigskip

\section{Classes of quantum codes for broadcast channels}
\label{class-of-codes}
While the tripartite operation $\mathcal{Z}^{ABC\to A'B'C'}$ between three space-like separated parties is generally a physically realizable CPTP map, a particular quantum communication setup may introduce additional time-like separation between the parties. Time-like separation refers to the causal ordering between events in spacetime, meaning one event can directly influence another because it occurs earlier in time within the causal light cone. This time-like separation may impose further constraints for the operation to qualify as a causally and physically realizable quantum code. Depending on the nature of the time-like separation and the corresponding constraint, the operation $\mathcal{Z}^{ABC\to A'B'C'}$ may serve as a code for either a quantum broadcast channel or a quantum multiple access channel. 
Our aim is to find the most generally realizable tripartite operation as a quantum code for transmission of quantum information over a two-receiver quantum broadcast channel. 
We will now explain the time-like separation induced by a quantum broadcast channel and how this condition naturally arises as a no-signalling condition for the tripartite operation.
We are interested in a code which helps transmit information over a broadcast channel $\mathcal{N}^{A'\to BC}$, whose input is controlled by Alice and receives are Bob and Charlie. 
The channel connects output $A'$ to the inputs $B$ and $C$ of the tripartite operation $\mathcal{Z}^{ABC\to A'B'C'}$; see Fig. \ref{fig:forward-assissted}. 
Therefore, the quantum broadcast channel $\mathcal{N}^{A'\to BC}$ represents a time-like separation because Bob’s and Charlie's inputs (or operations) depend on what they receive from channel whose input belongs to Alice. In other words, Alice's action causally influences Bob and Charlie.
Although the three players are spatially distant, the sequential nature imposes time-like separations as well. This separation means that information flows in a specific causal order, Alice to Bob and Charlie, and this influences how the code is modelled and constrained.
 Specifically, the actions of Bob and Charlie have to take place in the future of Alice's action.

\begin{figure}
\centering
\includegraphics[width=0.9\textwidth]{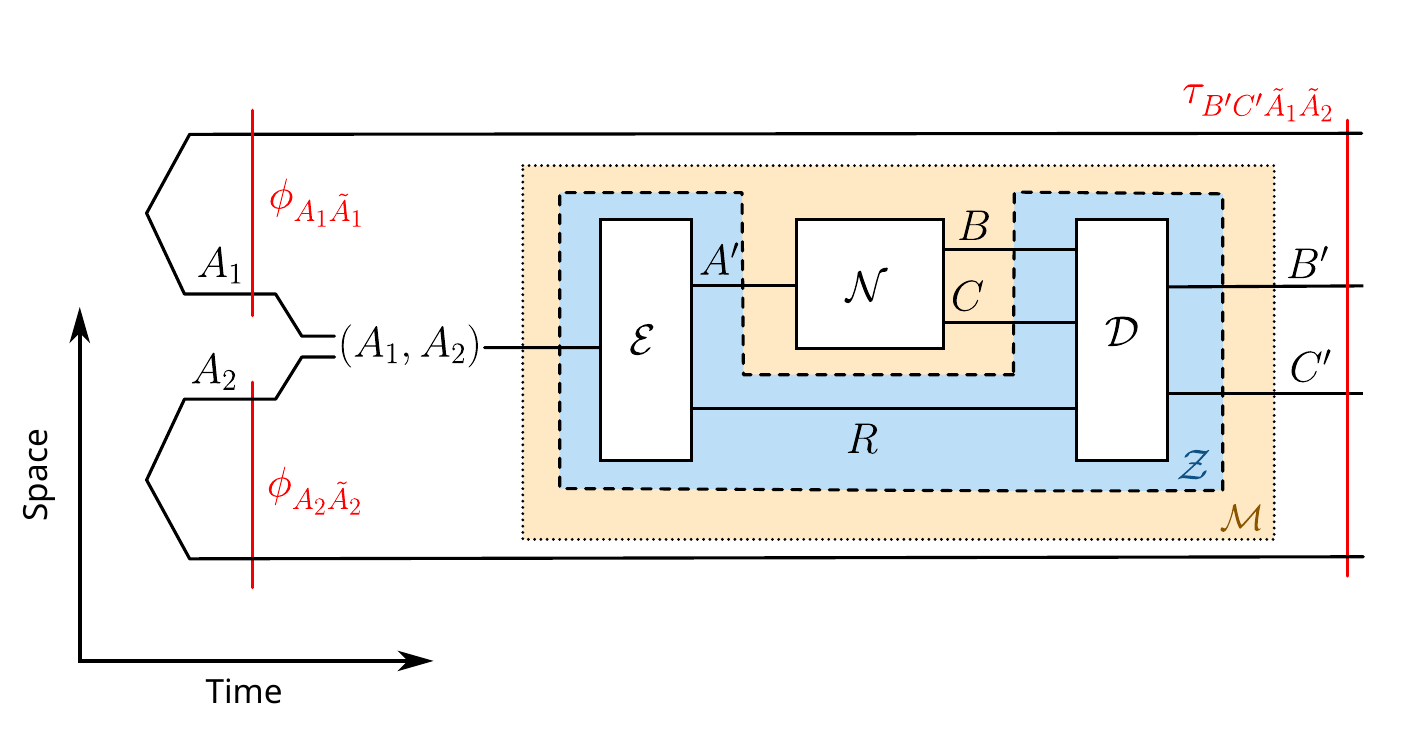}
\caption{Alice controls the inputs $A_1A_2$, such that she sends halves of her EPR pairs $\phi^{A_1\widetilde{A}_1}\phi^{A_2\widetilde{A}_2}$ through the encoder $\mathcal{E}$, and then over the channel $\mathcal{N}^{A'\to BC}$. The dashed (blue) box represents the most general forward-assisted quantum code, which can be interpreted as a superchannel converting the channel $\mathcal{N}^{A'\to BC}$ into the channel $\mathcal{M}^{A_1A_2\to B'C'}$, the latter being depicted by the dotted (orange) line box.}
\label{fig:forward-assissted}
\end{figure}

 Obviously, communication from Alice to Bob and Charlie cannot help Bob and Charlie to send a message (either classical or quantum) to Alice, so the operation is no-signalling from Bob and Charlie to Alice, i.e. $BC\not\to A$. Operations with this property are also referred to as semicausal in the literature \cite{PhysRevA.64.052309,TEggeling_2002}. In summary, the condition $BC\not\to A$ is required for the tripartite quantum channel $\mathcal{Z}^{ABC\to A'B'C'}$ to be a physically realizable quantum code. This in turn raises the question of what the most general classes of quantum codes for the quantum broadcast channel are, arising from the CPTP map $\mathcal{Z}^{ABC\to A'B'C'}$, where the only constraint is that Bob and Charlie cannot send messages to Alice, i.e. the no-signalling condition $BC\not\to A$. In the following, we show that this class corresponds to the class of semilocalizable operations \cite{TEggeling_2002}. We develop a proof along the same lines for the case of bipartite operations \cite[Sec. IV]{PhysRevA.74.012305}. (See also \cite[Theorem 7.1]{gour2024resourcesquantumworld}.) (This was originally proved with different techniques in \cite{TEggeling_2002}.):
\begin{lemma}
\label{forward-assisted-decompose}
    The class of quantum codes for the quantum broadcast channel $\mathcal{N}^{A'\to BC}$ satisfying the no-signaling condition $BC\not\to A$ is equivalent to a composition of local operations by Alice, a global operation by Bob and Charlie, and one-way quantum communication from Alice to Bob and Charlie. See Fig. \ref{fig:forward-assissted}
\end{lemma}
\begin{proof}
    Recall from Eq. \eqref{def-no-sig-bc-a} that the no-signalling condition $BC \not\to A$ for the Choi matrix $Z^{ABC \to A'B'C'}$ is expressed as follows:
\begin{align*}
 Z^{ABCA'} 
     = Z^{AA'} \otimes \frac{\mathds{1}^{BC}}{\abs{B}\abs{C}}.
\end{align*}
Let $\ket{\Psi}^{ABCA'B'C'P}$ be an unnormalized pure state, serving as purification of the Choi matrix $Z^{ABCA'B'C'}$. Obviously it constitutes a purification of the operator $Z^{ABCA'}$ as well. Since $Z^{ABCA'} 
     = Z^{AA'} \otimes \frac{\mathds{1}^{BC}}{\abs{B}\abs{C}}$, another purification of $Z^{ABCA'}$ can be written as:
     \begin{align*}
         \ket{\zeta}^{RABCA'} = \ket{\theta}^{AA'R} \otimes  \ket{\widebar{\phi}}^{B\widetilde{B}}\otimes \ket{\widebar{\phi}}^{C\widetilde{C}},
     \end{align*}
where $\ket{\theta}^{AA'R}$ is a purification of $Z^{AA'}/\abs{B}\abs{C}$, and $\ket{\widebar{\phi}}^{B\widetilde{B}} \otimes \ket{\widebar{\phi}}^{C\widetilde{C}}$ is a purification of the identity operator $\mathds{1}^{BC}$ (recall that we use $\widebar{\phi}$ to denote an unnormalized EPR pair). Since $Z^{ABC} = \mathds{1}^{ABC}$ (due to trace-preserving condition), we obtain $Z^{A} = \abs{B}\abs{C}\mathds{1}^{A}$. On the other hand, $\theta^{A} = Z^{A}/\abs{B}\abs{C}$ and therefore, $\theta^{A} = \mathds{1}^{A}$. Since unnormalized maximally entangled state $\ket{\widebar{\phi}}^{A\widetilde{A}}$ is another purification of $\theta^{A}$, this implies that there exists an (encoding) isometry $\mathcal{E}^{\widetilde{A}\to A'R}: \widetilde{A}\to A'\otimes R$ such that
\begin{align}
\label{encoding-isometry}
  \theta^{RAA'} =  \mathcal{E}^{\widetilde{A}\to A'R} \left(  \ketbra{\widebar{\phi}}^{A\widetilde{A}}\right).
\end{align}
As noted earlier, the unnormalized state $\ket{\Psi}^{ABCA'B'C'P}$ serves as a purification of both the Choi matrix $Z^{ABCA'B'C'}$ and $Z^{ABCA'}$. Since purifications are unique up to isometries on the purifying systems, there exists a (decoding) isometry $\mathcal{V}^{ R\widetilde{B}\widetilde{C} \to B'C'P }$ that relates the two purifications of $Z^{ABCA'}$ as follows:
\begin{align*}
   \ket{\Psi}^{ABCA'B'C'P} = \mathcal{V}^{ R\widetilde{B}\widetilde{C} \to B'C'P } \ket{\theta}^{AA'R} \otimes  \ket{\widebar{\phi}}^{B\widetilde{B}}\otimes \ket{\widebar{\phi}}^{C\widetilde{C}}.
\end{align*} 
Let $\mathcal{D}^{R\widetilde{B}\widetilde{C}\to  B'C'} (\cdot) = \Tr_{P} \circ \mathcal{V}^{ R\widetilde{B}\widetilde{C} \to B'C'P } (\cdot) $. We have:
\begin{align*}
    Z^{ABCA'B'C'} &= \mathcal{D}^{R\widetilde{B}\widetilde{C}\to  B'C'} \left( \theta^{AA'R} \otimes  \widebar{\phi}^{B\widetilde{B}} \otimes \widebar{\phi}^{C\widetilde{C}} \right) \\
    &= \mathcal{D}^{R\widetilde{B}\widetilde{C}\to  B'C'} \left(  \mathcal{E}^{\widetilde{A}\to A'R}  \big(\widebar{\phi}^{A\widetilde{A}}\big) \otimes \widebar{\phi}^{B\widetilde{B}} \otimes \widebar{\phi}^{C\widetilde{C}}\right)\\
    &= \mathcal{D}^{R\widetilde{B}\widetilde{C}\to  B'C'} \circ  \mathcal{E}^{\widetilde{A}\to A'R}  \left(\widebar{\phi}^{A\widetilde{A}} \otimes \widebar{\phi}^{B\widetilde{B}} \otimes \widebar{\phi}^{C\widetilde{C}}\right),
\end{align*}
where the first line follows from $\Tr_{P}\Psi^{ABCA'B'C'P} = Z^{ABCA'B'C'}$ and in the second line we use Eq. \eqref{encoding-isometry}.
Note that Alice's action $\mathcal{E}^{\widetilde{A}\to A'R}$ is an isometry and the joint channel of Bob and Charlie $\mathcal{D}^{R\widetilde{B}\widetilde{C}\to  B'C'}$ is a CPTP map.
\end{proof}

 This class of codes is referred to as \textit{forward-assisted codes} for point-to-point channels in \cite{leung-matthews}. However, our model simplifies the framework of forward-assisted codes presented in \cite{leung-matthews} (see Fig.~1 therein). Specifically, we observe that the operation $\mathcal{F}^{Q \to R}$ introduced in \cite[Fig. 1]{leung-matthews} is unnecessary, as it can be fully absorbed into the encoding isometry $\mathcal{E}$. Thus, the two classes of codes are essentially equivalent.
 We shall refer to these codes as forward-assisted codes for quantum broadcast channels, aligning with the terminology previously used in \cite{leung-matthews}. In the following, we formally define the class of forward-assisted codes for quantum broadcast channels:
\begin{definition}
    A forward-assisted code for the quantum broadcast channel $\mathcal{N}^{A' \to BC}$, as illustrated in Fig.~\ref{fig:forward-assissted}, consists of Alice's encoding isometry $\mathcal{E}^{A \to A'R}$ and Bob and Charlie's joint decoding channel $\mathcal{D}^{RBC \to B'C'}$. The states to be transmitted by Alice live in the Hilbert spaces $A_1$ and $A_2$, where $\abs{A_1} = \abs{B'} = r_1$ and $\abs{A_2} = \abs{C'} = r_2$. Alice transmits the system $A'$ through the quantum broadcast channel $\mathcal{N}^{A' \to BC}$, while the system $R$ is sent through an identity channel. This corresponds to the tripartite operation $\mathcal{Z}^{ABC\to A'B'C'} = \mathcal{D}^{RBC \to B'C'} \circ \mathcal{E}^{A \to A'R}$ outlined with dashes in Fig. \ref{fig:forward-assissted}. We refer to the pair $(r_1, r_2)$ as the size of the code.
\end{definition}

It is important to note that the class of forward-assisted codes does not necessarily represent the most general class of codes; most importantly, forward-assisted codes do not allow for feedback and adaptive encoding between channel uses.

The encoding isometry $\mathcal{E}^{A \to A'R}$, where $A = A_1 \otimes A_2$, acting on halves of the EPR pairs $\phi^{A_1\widetilde{A}_1}\phi^{A_2\widetilde{A}_2}$ results in an \textit{average channel input state} defined as follows:
\begin{align} 
\label{avg-input}
\rho^{A'} = \Tr_{\widetilde{A}_1\widetilde{A}_2 R} \mathcal{E}^{A \to A'R} \phi^{A_1\widetilde{A}_1}\phi^{A_2\widetilde{A}_2}. 
\end{align} 
This state is particularly useful as it serves as a variable in our SDPs. The Choi matrix of the tripartite operation $\mathcal{Z}^{ABC\to A'B'C'} = \mathcal{D}^{RBC \to B'C'} \circ \mathcal{E}^{A \to A'R}$, i.e. the code, can be written as follows in terms of the Choi matrices of the encoding isometry $E^{A A'R}$ and the decodiong channel $D^{RBCB'C'}$:
\begin{align*}
    Z^{ABC A'B'C'} = \Tr_{R} D^{RBCB'C'} (E^{A A'R})^{T_R}.
\end{align*}

The process of encoding and decoding results in an overall CPTP map $\mathcal{M}^{A_1A_2 \to B'C'} = \mathcal{D}^{RBC \to B'C'} \circ \mathcal{N}^{A' \to BC} \circ \mathcal{E}^{A \to A'R}$, where $\abs{A_1} = \abs{B'}=r_1$ and $\abs{A_2} = \abs{C'}= r_2$. To put it another way, the forward-assisted code $\mathcal{Z}^{ABC\to A'B'C'}$ constitutes a linear map from quantum broadcast channels $\mathcal{N}^{A' \to BC}$ to quantum CPTP maps $\mathcal{M}^{A_1A_2 \to B'C'}$ \cite{Chiribella_2008}. The Choi matrix of the induced channel $\mathcal{M}^{A_1A_2 \to B'C'}$ can be expressed as follows in terms of the Choi matrices of its constituent channels:
\begin{align*}
    M^{AB'C'} & = \Tr_{A'RBC} D^{RBCB'C'} (N^{A'BC})^{T_{BC}} (E^{A A'R})^{T_{A'R}} \\
    & = \Tr_{A'RBC} D^{RBCB'C'} (N^{A'BC})^{T_{A'BC}} (E^{A A'R})^{T_{R}}\\
    & = \Tr_{A'BC} Z^{ABCA'B'C'} (N^{A'BC})^{T} 
\end{align*}
Note that the average channel input defined in Eq. \eqref{avg-input} can be expressed using the Choi matrix $Z^{ABCA'B'C'}$ as follows:
\begin{align}
\label{avg-input-2}
   \rho^{A'} = \Tr_{ABCB'C'} Z^{ABCA'B'C'} \frac{\mathds{1}^{A}}{\abs{A}}\frac{\mathds{1}^{B}}{\abs{B}}\frac{\mathds{1}^{C}}{\abs{C}}.
\end{align}
This emphasizes that the average input state depends only on the encoding isometry $\mathcal{E}^{A\to A'R}$, i.e., it does not depend on any decoding map or the particular quantum broadcast channel used for the transmission of information.

For the resulting quantum channel $\mathcal{M}^{A_1A_2\to B'C'}$ with $\abs{A_1}=\abs{B'}$, $\abs{A_2}=\abs{C'}$, its channel fidelity is defined as follows:
\begin{align*}
    \Tr \phi^{\widetilde{A}_1B'}\phi^{\widetilde{A}_2C'} \mathcal{M}^{A_1A_2\to B'C'}(\phi^{A_1\widetilde{A}_1}\phi^{A_2\widetilde{A}_2}).
\end{align*}
(Note that here and elsewhere, we freely switch between systems with and without tildes for ease of notation.)
When Alice's inputs are halves of the EPR pairs $\phi^{A_1\widetilde{A}_1}$ and $\phi^{A_2\widetilde{A}_2}$, the overall effect of the encoded transmission yields a state $\tau^{\widetilde{A}_1\widetilde{A}_2B'C'}$, as shown in Fig.~\ref{fig:forward-assissted}. The channel fidelity of $\mathcal{M}^{\widetilde{A}_1\widetilde{A}_2\to B'C'}$ is the entanglement fidelity of the state $\tau^{\widetilde{A}_1\widetilde{A}_2B'C'}$, i.e.
\begin{align*}
    \Tr \phi^{\widetilde{A}_1B'}\phi^{A_2\widetilde{A}_2C'}\tau^{\widetilde{A}_1\widetilde{A}_2B'C'},
\end{align*}
and this is called the channel fidelity of the forward-assisted code. In terms of the Choi matrices, the channel fidelity of the code is expressed as follows: 
\begin{align}
\label{channel-fidelity-express}
    \frac{1}{r_1r_2} \Tr \phi^{A_1B'}\phi^{A_2C'} M^{A_1A_2B'C'}
    = \frac{1}{r_1r_2} \Tr \phi^{A_1B'}\phi^{A_2C'}Z^{ABCA'B'C'} (N^{A'BC})^{T}.
\end{align}

\medskip
 The class of forward-assisted codes, being the most general class of codes for a single use of a quantum broadcast channel, is too powerful to be practically interesting. First, Alice can use the register $R$ to transmit information to Bob and Charlie without even using the quantum broadcast channel. Moreover, the decoding channel $\mathcal{D}^{RBC \to B'C'}$ is a joint operation under the control of both Bob and Charlie, even though these parties are physically separated.
In contrast, the most restricted class of codes, known as \textbf{unassisted (UA) codes}, allows neither forward assistance nor cooperation between the receivers. Here, the operation $\mathcal{Z}^{ABC \to A'B'C'}$ must take the product form 
\begin{align*}
    \mathcal{Z}^{ABC \to A'B'C'} =  \mathcal{D}^{C \to C'}_c \mathcal{D}^{B \to B'}_b \mathcal{E}^{A \to A'},
\end{align*}
where the operations $\mathcal{E}^{A \to A'}$, $\mathcal{D}^{B \to B'}_b$, and $\mathcal{D}^{C \to C'}_c$ are performed locally by Alice, Bob, and Charlie, respectively, and are otherwise arbitrary.
Between these extremes, we aim to identify genuinely interesting classes of quantum codes. Specifically, we focus on classes that are both practically relevant and lend themselves to meaningful optimization. As mentioned earlier, forward-assisted codes are constrained by the no-signalling condition from Bob and Charlie to Alice, i.e., $BC \not\to A$.
We propose two new classes of codes that satisfy these criteria: NS codes and PPT-preserving codes. NS codes form a subclass of forward-assisted codes, where no party can signal to any other party. According to Lemma~\ref{no-signalling-enough}, this means the conditions $A \not\to BC$, $B \not\to AC$, and $C \not\to AB$ must hold. While these conditions imply $BC \not\to A$, we emphasize the former conditions for reasons that will become clear later.
A particularly promising class for tractable optimization is the subclass of PPT-preserving codes. In these codes, the tripartite operation $\mathcal{Z}^{ABC \to A'B'C'}$ preserves the positive partial transpose (PPT) property across all three bipartitions. Additionally, we will study the codes which are both NS and PPT-preserving.
We might occasionally be interested in the class of \textbf{entanglement-assisted (EA) codes}. Here, the parties are provided unlimited shared entanglement. Since entanglement is a no-signalling resource, the EA codes are a subclass of NS codes.

We are now in a position to formally define, for any class of codes $\Omega$, where $\Omega \in \{\text{UN},\text{EA}, \text{PPT}, \text{NS}, \text{PPT} \cap \text{NS}\}$, the maximal channel fidelities as well as one-shot and asymptotic quantum capacities of the quantum broadcast channel $\mathcal{N}^{A' \to BC}$.
\begin{definition}
    Let $\mathcal{M}^{A_1A_2\to B'C'}_{\Omega}$ be the channel resulting from the composition of the tripartite operation $\mathcal{Z}^{A_1A_2BC\to A'B'C'}$ corresponding to the code $\Omega$ and the quantum broadcast channel $\mathcal{N}^{A'\to BC}$. The maximum channel fidelity of the broadcast channel $\mathcal{N}^{A'\to BC}$ assisted by the code $\Omega$ is defined as 
    \begin{align*}
        F^{\Omega}(\mathcal{N}, r_1,r_2)\coloneqq
        \sup_{\Omega}\Tr{\phi^{\widetilde{A}_1B'} \phi^{\widetilde{A}_2C'}\mathcal{M}^{A_1A_2\to B'C'}_{\Omega}(\phi^{\widetilde{A}_1A_1} \phi^{\widetilde{A}_2A_2})},
    \end{align*}
where $r_1=\abs{A_1}=|\widetilde{A}_1|=\abs{B'}$ and $r_2=\abs{A_2}=|\widetilde{A}_2|=\abs{C'}$. The pair $(r_1,r_2)$ is called the size of the code $\Omega$ and the supremum is taken over the codes $\Omega \in \{\text{UN}, \text{EA}, \text{PPT}, \text{NS}, \text{PPT}\cap \text{NS}\}$.
\end{definition}

 The following inequalities illustrate the relationships between the maximal channel fidelities achieved by the various classes of codes discussed above:
\begin{align}
    F^{\text{UA}}(\mathcal{N}, r_1,r_2) &\leq F^{\text{EA}}(\mathcal{N}, r_1,r_2) \leq F^{\text{NS}}(\mathcal{N}, r_1,r_2), \\
    F^{\text{UA}}(\mathcal{N}, r_1,r_2) &\leq F^{\text{NS}\cap \text{PPT}}(\mathcal{N}, r_1,r_2) \leq  F^{\text{PPT}}(\mathcal{N}, r_1,r_2),\\
    F^{\text{UA}}(\mathcal{N}, r_1,r_2) &\leq F^{\text{NS}\cap \text{PPT}}(\mathcal{N}, r_1,r_2) \leq  F^{\text{NS}}(\mathcal{N}, r_1,r_2).
\end{align}

We next define the one-shot and asymptotic capacities. For the asymptotic capacity, we assume that the channel is memoryless, with its operation for $n$ uses given by $\mathcal{N}^{\otimes n}$.

\begin{definition}
\label{one-shot-def}
 Let $\varepsilon\geq 0$. For a quantum broadcast channel $\mathcal{N}^{A'\to BC}$, a rate pair $(R_1\coloneqq\log r_1,R_2\coloneqq \log r_2)$ is called $\varepsilon$-one-shot achievable if $F^{\Omega}(\mathcal{N},r_1,r_2)\geq 1-\varepsilon$. The $\varepsilon$-error quantum capacity region is defined as the closure of the set of all $\varepsilon$-one-shot achievable rate pairs $(R_1,R_2)$. The corresponding asymptotic capacity region is defined as $\{(R_1,R_2): \lim_{n\to\infty} F^{\Omega}(\mathcal{N}^{\otimes n},2^{nR_1},2^{nR_2})=1\}$.
 Sometimes we are interested in the $\varepsilon$-one-shot sum-capacity defined as 
 \begin{align*}
     Q^{\Omega}_{1}(\mathcal{N},\varepsilon)=\max \{R_1+R_2: (R_1,R_2) \, \text{is $\varepsilon$-one-shot achievable} \}.
 \end{align*}
The corresponding asymptotic sum-capacity is defined as
 \begin{align}
 \label{asym-sum-capacity}
     Q^{\Omega}(\mathcal{N})=\lim_{\varepsilon\to 0}\lim_{n\to\infty}\frac{1}{n}Q^{\Omega}_{1}(\mathcal{N}^{\otimes n},\varepsilon).
 \end{align}
\end{definition}

As discussed earlier, our motivation for considering NS and/or PPT-preserving codes arises from their compatibility with efficient optimization techniques, enabling the search for good codes. In particular, the next section demonstrates how the optimal channel fidelity of forward-assisted codes, which are either no-signaling, PPT-preserving, or both, can be expressed as semidefinite programs (SDPs) \cite{watrous_lecture_notes,boyd2004convex}. For an introduction to SDPs. See also \cite[Sec. 2.4]{khatri-wilde-book}. SDPs possess several notable properties: they facilitate efficient numerical computation, any feasible solution to the dual SDP serves as an upper bound for the primal problem, and when strong duality holds (as will be explained later), the dual solution directly provides the optimal value.

\bigskip

\section{Semidefinite programs for PPT-preserving and no-signalling codes}
\label{sdp-fidelity}
We have seen that the general class of forward-assisted codes of size $(r_1,r_2)$ corresponds to CPTP maps $\mathcal{Z}^{ABC\to A'B'C'}$ which are no-signalling from Bob and Charlie to Alice, i.e. $BC\not\to A$. These necessary yet minimal requirements are manifested in the Choi matrix of the tripartite operation $\mathcal{Z}^{ABC\to A'B'C'}$ as follows:

\begin{align}
\label{CP}
  \textbf{CP}\,\,\, &: \,\, Z^{ABCA'B'C'}\geq 0,\\
    \label{TP}
   \textbf{TP}\,\,\, &: \,\, \Tr_{A'B'C'}=\mathds{1}^{ABC},\\
    \label{BCnotA}
  \textbf{NS}\,\,(BC\not\to A)\,\,\, &: \,\,  \Tr_{B'C'}Z^{ABCA'B'C'} = \frac{\mathds{1}^{BC}}{\abs{B}\abs{C}} \Tr_{BCB'C'}Z^{ABCA'B'C'}.
\end{align}
Here, Eqs. \eqref{CP} and \eqref{TP} correspond to the operation being completely positive and trace-preserving, respectively. The equality in Eq. \eqref{BCnotA} represents the constraint that the operation is causal, i.e. Bob and Charlie, even when considered together as a single entity, cannot send messages to Alice (see Eq. \eqref{def-no-sig-bc-a}).

The forward-assisted code is no-signalling if and only if no party can communicate messages to any other party. From Lemma \ref{no-signalling-enough}, this condition is equivalent to the following three constraints on the Choi matrix:

\begin{align}
\label{anottobc}
    \textbf{NS}\,\,(A\not\to BC)\,&:\,\, \Tr_{A'} Z = \mathds{1}^{A_1A_2}/r_1 r_2 \otimes \Tr_{A'A_1A_2}Z^{ABCA'B'C'},\\
    \label{bnottoac}
     \textbf{NS}\,\,(B\not\to AC)\,&:\,\, \Tr_{B'} Z = \mathds{1}^{B}/\abs{B} \otimes \Tr_{B'B}Z^{ABCA'B'C'},\\
     \label{cnottoab}
     \textbf{NS}\,\,(C\not\to AB)\,&:\,\, \Tr_{C'} Z = \mathds{1}^{C}/\abs{C} \otimes \Tr_{C'C}Z^{ABCA'B'C'}.
     \end{align}
    Note that, according to Lemma \ref{no-signalling-enough}, these three conditions imply Eq. \eqref{BCnotA}. However, we will use the constraint in Eq.~\eqref{BCnotA} to simplify our SDPs, as will be demonstrated shortly.

A forward-assisted code is PPT-preserving if and only if the tripartite operation remains PPT-preserving across all three parties. This condition can be expressed in terms of the Choi matrix as follows:

     \begin{align}
     \label{ppt-a}\textbf{PPT}\textbf{\textsubscript{A}}\, &:\,\, (Z^{ABCA'B'C'})^{T_{AA'}}\geq 0,\\  
     \label{ppt-b}\textbf{PPT}\textbf{\textsubscript{B}}\, &:\,\, (Z^{ABCA'B'C'})^{T_{BB'}}\geq 0,\\
     \label{ppt-c}
     \textbf{PPT}\textbf{\textsubscript{C}}\, &:\,\, (Z^{ABCA'B'C'})^{T_{CC'}}\geq 0.
\end{align}

Since these codes are fully characterized as conditions on the Choi matrix of the tripartite channel, we aim to express the channel fidelity directly in terms of the Choi matrix. This allows us to maximize the fidelity under the given conditions, formulating the problem as a semidefinite program (SDP). As noted earlier (Eq. \eqref{channel-fidelity-express}), the channel fidelity is given by
\begin{align}
    \label{channel-f-express}
    f_c = \frac{1}{r_1r_2} \Tr \phi^{A_1B'}\phi^{A_2C'}Z^{ABCA'B'C'} (N^{A'BC})^{T}.
\end{align}
Our goal in this section is to maximize $f_c$ subject to Eqs. \eqref{CP}-\eqref{BCnotA}, together with the additional constraints Eqs. \eqref{anottobc}-\eqref{ppt-c}, as appropriate. 

A crucial step in finding the maximal channel fidelity as a reasonably simple SDP is to consider a highly symmetric version of the Choi matrix, known as the \textit{twirled Choi matrix}, where the Choi matrix is twirled according to some unitary group. A well-known result, foundational to many benchmarking protocols in quantum computation (see, e.g., \cite{onorati2024}), states that multiplying the inputs and outputs of a channel by certain unitaries in a specific way does not alter the channel fidelity \cite{PhysRevA.60.1888}. Here, we apply this result to demonstrate that the twirled Choi matrix can be used without any loss of generality.

We are particularly interested in multipartite operators, specifically focusing on closed groups of unitaries of the form $\left(U^{A_1} \otimes U^{A_2}\right)$, where $U^{A_1}$ and $U^{A_2}$ are unitary matrices in $A_1$ and $A_2$, respectively. For simplicity, we have streamlined the notation, as we will consistently refer to this specific unitary representation in the following discussion. Being a closed subgroup of the unitary group, this group is necessarily compact and thus admits a unique Haar measure, denoted by $\mu(U)$. This property allows us to define the concept of a twirl.
Let $U^\ast$ represent the complex conjugate of $U$.
By using the transpose trick, i.e. $U^{A_1} \ket{\phi}^{A_1 B'} =(U^{B'})^{T}\ket{\phi}^{A_1 B'}, U^{A_2}\ket{\phi}^{A_2 C'}= (U^{C'})^{T}\ket{\phi}^{A_2 C'}$, we have
\begin{align*}
    (U^{A_1})^{T} (U^{B'})^{\dagger}\ket{\phi}^{A_1 B'}=\ket{\phi}^{A_1 B'},\\
     (U^{A_2})^{T} (U^{C'})^{\dagger}\ket{\phi}^{A_2 C'}=\ket{\phi}^{A_2 C'},
\end{align*}
Consequently,
\begin{align*}
   (U^{A_1})^{T} (U^{B'})^{\dagger}(U^{A_2})^{T}\ (U^{C'})^{\dagger}\ket{\phi}^{A_1 B'} \ket{\phi}^{A_2 C'}=\ket{\phi}^{A_1 B'}\ket{\phi}^{A_2 C'},
\end{align*}
or 
\begin{align}
\label{transposed}
    (U^{A_1})^{T} (U^{B'})^{\dagger}(U^{A_2})^{T}\ (U^{C'})^{\dagger}(\phi^{A_1 B'} \phi^{A_2 C'})(U^{A_1})^{\ast} U^{B'}(U^{A_2})^{\ast}\ U^{C'}=\phi^{A_1 B'} \phi^{A_2 C'}.
\end{align}
For conciseness, we denote
\begin{align}
\label{concise-U}
  \widehat{U}^{A_1A_2B'C'}\coloneqq (U^{A_1})^{T} (U^{B'})^{\dagger}(U^{A_2})^{T}\ (U^{C'})^{\dagger}.
\end{align}
From this, the channel fidelity Eq. \eqref{channel-f-express} satisfies 
\begin{align*}
      f_c
    &=
    \frac{1}{r_1r_2}
    \Tr\big\{(\phi^{A_1B'} \phi^{A_2C'})Z^{A_1A_2BCA'B'C'}(N^{A'BC})^{T}
    \big\}\\
    &=\frac{1}{r_1r_2}\int_{U}d\mu(U)
    \Tr\big\{\widehat{U}^{A_1A_2B'C'}(\phi^{A_1B'} \phi^{A_2C'})(\widehat{U}^{A_1A_2B'C'})^{\dagger}Z^{A_1A_2BCA'B'C'}(N^{A'BC})^{T}
    \big\}\\
    &=\frac{1}{r_1r_2}\int_{U}d\mu(U)
    \Tr\big\{(\phi^{A_1B'} \phi^{A_2C'})(\widehat{U}^{A_1A_2B'C'})^{\dagger}Z^{A_1A_2BCA'B'C'}\widehat{U}^{A_1A_2B'C'}(N^{A'BC})^{T}
    \big\}\\
       &=\frac{1}{r_1r_2}
    \Tr\big\{(\phi^{A_1B'} \phi^{A_2C'})\widebar{Z}^{A_1A_2BCA'B'C'}(N^{A'BC})^{T}
    \big\}.
\end{align*}
The first line follows from the definition of the channel fidelity in Eq.\eqref{channel-f-express}, the second line from Eq.~\eqref{transposed}, and the third line from the cyclic property, the linearity of the trace, and the definition of the Haar measure. The operator $\widebar{Z}^{A_1A_2BCA'B'C'}$ in the final line is referred to as the twirled Choi matrix, which is given by:
\begin{align}
\label{twirled-choi-1}
  \widebar{Z}^{A_1A_2BCA'B'C'}  =
  \int_{U}d\mu(U)
  (\widehat{U}^{A_1A_2B'C'})^{\dagger}Z^{A_1A_2BCA'B'C'}\widehat{U}^{A_1A_2B'C'}.
\end{align}

While it is relatively straightforward to see that the twirled Choi matrix achieves the same channel fidelity, it is necessary to ensure that other properties of the Choi matrix, namely, the no-signalling and PPT-preserving properties, are maintained after twirling. To establish this, we begin by identifying the tripartite CPTP map corresponding to the twirled Choi matrix. We refer to this tripartite operation as the twirled tripartite operation and denote it by $\widebar{\mathcal{Z}}^{ABC \to A'B'C'}$. 
We will show that applying specific unitaries before and after the action of the tripartite channel $\mathcal{Z}^{A_1A_2BC \to A'B'C'}$ and averaging over all such unitaries yields the twirled version of this operator. Let us consider the Choi matrix with an instance of the unitary $U^{A_1A_1B'C'}$ given by Eq. \eqref{concise-U} as follows:
\begin{align*}
    (\widehat{U}^{A_1A_2B'C'})^{\dagger}Z^{A_1A_2BCA'B'C'}\widehat{U}^{A_1A_2B'C'} = 
    (U^{A_1})^{\ast} U^{B'} (U^{A_2})^{\ast} U^{C'} Z^{A_1A_2BCA'B'C'}
    (U^{A_1})^{T} (U^{B'})^{\dagger} (U^{A_2})^{T} (U^{C'})^{\dagger}.
\end{align*}

Now, consider the action of the unitary $(U^{A_1})^{\ast} (U^{A_2})^{\ast}$ on the Choi matrix. Using the transpose trick, these unitaries can be taken inside the action of the tripartite operation $\mathcal{Z}^{A_1A_2BC \to A'B'C'}$. This implies that these unitaries are applied to the inputs $A_1A_2$ prior to the tripartite channel acting on the state. On the other hand, the action of the unitaries $U^{B'} U^{C'}$ can be interpreted as first applying the tripartite operation, followed by applying these unitaries to the outputs $B'C'$. Thus, the twirled Choi matrix in Eq.~\eqref{twirled-choi-1} corresponds to the following modified tripartite operation:
\begin{align}
\label{twirled-code}
    \widebar{\mathcal{Z}}^{A_1A_2BC\to A' B'C'}(\cdot)=\int_{U}d\mu(U) U^{B'C'}\left( \mathcal{Z}^{A_1A_2BC\to A' B'C'}\left[(U^{A_1A_2})^{\dagger}(\cdot) U^{A_1A_2}\right]\right)(U^{B'C'})^{\dagger},
\end{align}
where we defined $U^{A_1A_2}\coloneqq U^{A_1} U^{A_2}$ and $U^{B'C'}\coloneqq U^{B'}U^{C'}$.

This twirled tripartite channel $\widebar{\mathcal{Z}}^{A_1A_2BC\to A' B'C'}$ can be implemented as follows: Alice shares a random variable with Bob and Charlie identifying a unitary $U^{A_1A_2}$ drawn according to Haar measure $\mu(U)$. Before the communication begins, Alice applies $(U^{A_1A_2})^{\dagger}$ to her inputs $A_1A_2$. Alice, Bob and Charlie then use the forward-assisted code according to $\mathcal{Z}^{A_1A_2BC\to A' B'C'}$. Subsequently, Bob and Charlie apply the unitaries $B^{B'}$ and $U^{C'}$, respectively, inverting Alice's operation on the input.
This means that the operation $\widebar{\mathcal{Z}}^{A_1A_2BC\to A' B'C'}$ can be implemented by using the operation $\mathcal{Z}^{A_1A_2BC\to A' B'C'}$ together with local operations and shared classical randomness. Therefore, if the operation $\mathcal{Z}^{A_1A_2BC\to A' B'C'}$ was non-signalling or PPT-preserving, $\overline{\mathcal{Z}}^{A_1A_2BC\to A' B'C'}$ will inherit these properties. This argument follows \cite{leung-matthews}, but \cite{8012535} provides slightly different idea, but they are essentially the same. This says that for a given broadcast channel $\mathcal{N}^{A'\to BC}$, as far as channel fidelity is concerned, the code operations $\overline{\mathcal{Z}}^{A_1A_2BC\to A' B'C'}$ and $\mathcal{Z}^{A_1A_2BC\to A' B'C'}$, subject to no-signalling and PPT-preserving constraints, yield identical results. However, the twirled operation exhibits symmetry properties that are particularly advantageous for simplifying the computation of channel fidelity. 

\bigskip

In order to establish the symmetry of the twirled Choi matrix Eq. \eqref{twirled-choi-1}, we exploit the projection formula given by Eq. \eqref{projection}. For that, we need to identify sub-representations of the unitary group $(\widehat{U}^{A_1A_2B'C'})$ on four-fold tensor product space. We identify the following projectors onto these four subspaces:
\begin{align*}
      \Pi_1^{A_1A_2B'C'} =& \phi^{A_1B'}\otimes \phi^{A_2C'},\\
   \Pi_2^{A_1A_2B'C'} = &\phi^{A_1B'}\otimes (\mathds{1}^{A_2C'}-\phi^{A_2C'}),\\
  \Pi_3^{A_1A_2B'C'} =  &(\mathds{1}^{A_1B'}-\phi^{A_1B'})\otimes \phi^{A_2C'},\\
   \Pi_4^{A_1A_2B'C'} = &(\mathds{1}^{A_1B'}-\phi^{A_1B'})\otimes (\mathds{1}^{A_2C'}-\phi^{A_2C'}).
\end{align*}
It is straightforward to check that $\Pi_1^{A_1A_2B'C'}+\Pi_2^{A_1A_2B'C'}+\Pi_3^{A_1A_2B'C'}+\Pi_4^{A_1A_2B'C'} = \mathds{1}^{A_1A_2B'C'}$. We can now use the projection formula to write the twirled operator $\widebar{Z}^{A_1A_2BCA'B'C'}$ as follows:
\begin{align}
\nonumber
    \widebar{Z}^{ABCA'B'C'} =&
r_1r_2\bigg( \phi^{A_1B'}   \phi^{A_2C'}    E_1^{A'BC} +
\phi^{A_1B'}   (\mathds{1}^{A_2C'}-\phi^{A_2C'})    E_2^{A'BC}\\
\label{twirled-choi}
&+(\mathds{1}^{A_1B'}-\phi^{A_1B'})   \phi^{A_2C'}    E_3^{A'BC}
+(\mathds{1}^{A_1B'}-\phi^{A_1B'})   (\mathds{1}^{A_2C'}-\phi^{A_2C'})    E_4^{A'BC}\bigg),
\end{align}
for some operators $E_1^{A'BC},E_2^{A'BC},E_3^{A'BC}$ and $E_4^{A'BC}$. As noted earlier, here and throughout the paper, we define $\abs{A_1}\coloneqq r_1,\abs{A_2}\coloneqq r_2$. When we write the operators $\{E_i\}_i$ superscripted with only a subset of the systems $\{A',B,C\}$, we refer to the partial trace of the operator over the respective subsystem, for example $E_1^{A'} = \Tr_{BC}E_1^{A'BC}$.
The average channel input for the twirled forward-assisted code $\widebar{\mathcal{Z}}^{ABC\to A'B'C'}$ is obtained by substituting the twirled Choi matrix from Eq. \eqref{twirled-choi} into the corresponding expression provided in Eq. \eqref{avg-input-2}, as follows:

\begin{align}
\nonumber
    \rho^{A'} &= \Tr_{ABCB'C'}\left\{\widebar{Z}^{ABCA'B'C'}\frac{\mathds{1}^{ABC}}{r_1r_2\abs{B}\abs{C}}\right\}\\
\label{avg-twirled}
    &=\frac{1}{\abs{B}\abs{C}}\left(E_1^{A'} + (r_2^2-1)E_2^{A'} + (r_1^2-1)E_3^{A'} +(r_1^2-1)(r_2^2-1)E_4^{A'} \right).
\end{align}

We are now ready to derive an expression for the quantum broadcast channel fidelity in terms of the twirled Choi matrix and the average channel input. This forms the foundation for the subsequent results presented in this work.

\begin{theorem}
\label{full-fledged}
    For a quantum broadcast channel $\mathcal{N}^{A'\to BC}$, there is a forward-assisted code (See Fig. \ref{forward-assisted-decompose}) of size $(r_1,r_2)$, average channel input $\rho^{A'}$ and channel fidelity $f_c$, which is PPT-preserving and/or no-signalling if and only if the following SDP has a feasible solution
    \end{theorem}
    \begin{align}
    \label{channel-fidelity}
         &f_c =
    \Tr E_1^{A'BC}(N^{A'BC})^{T}\\
    \label{cp-conditions}
     & E^{A'BC}_i\geq 0, \quad \text{for} \quad i=1,2,3,\\
    \label{bcnottoa}
    & E_{1}^{A'BC}+ (r_2^2  - 1)E_{2}^{A'BC} + (r_1^2  - 1)E_{3}^{A'BC} \leq  \rho^{A'}\mathds{1}^{BC}.\\
    \label{tp+ns}
     \textbf{NS}\,\,(A\not\to BC):&\,\, E_i^{BC}=\frac{\mathds{1}^{BC}}{r_1^2r_2^2}, \quad i=1,2,3,\\
    \label{BnototAC-one}
     \textbf{NS}\,\,(B\not\to AC):&\,\, E_1^{A' BC} + (r_1^2 - 1)E_3^{A' BC} =  \big(E_1^{A' C} + (r_1^2 - 1)E_3^{A' C}\big) \otimes \frac{\mathds{1}^B}{\abs{B}},\\
    \label{CnototAB-one}
     \textbf{NS}\,\,(C\not\to AB):&\,\, E_1^{A' BC} + (r_2^2 - 1)E_2^{A' BC} =  \big(E_1^{A' B} + (r_2^2 - 1)E_2^{A' B}\big) \otimes \frac{\mathds{1}^C}{\abs{C}}.
    \end{align}
    \begin{align}
    \label{ppt-a-theorem}
    \textbf{PPT}\textbf{\textsubscript{A}}:&
    \begin{cases}
    &     - \frac{1}{r_1r_2}\rho^{A'}\mathds{1}^{BC}  - (1+\frac{1}{r_2}+\frac{1}{r_1})E_1^{T_{BC}}\leq \frac{r_2^2-1}{r_2}E_2^{T_{BC}}
  + \frac{r_1^2-1}{r_1}E_3^{T_{BC}} \leq \frac{1}{r_1r_2}\rho^{A'}\mathds{1}^{BC} + (1-\frac{1}{r_2}-\frac{1}{r_1})E_1^{T_{BC}}\\
& - \frac{1}{r_1r_2}\rho^{A'}\mathds{1}^{BC} +  (1-\frac{1}{r_2}+\frac{1}{r_1})E_1^{T_{BC}}\leq   \frac{r_2^2-1}{r_2}E_2^{T_{BC}}
  - \frac{r_1^2-1}{r_1}E_3^{T_{BC}} \leq \frac{1}{r_1r_2}\rho^{A'}\mathds{1}^{BC} - (1+\frac{1}{r_2}-\frac{1}{r_1})E_1^{T_{BC}}.
    \end{cases}
\\
    \label{ppt-b-theorem}
    \textbf{PPT}\textbf{\textsubscript{B}}:&
    \begin{cases}
        \hspace{3.6cm} -(r_1 -1)E_3^{T_B} \leq E_1^{T_B} \leq (r_1 +1)E_3^{T_B},\\
- (  \rho^{A'}\mathds{1}^{BC} - (r_1 ^2 - 1)E_3^{T_B} - E_1^{T_B})    \leq r_1(r_2 ^2-1) E_2^{T_B}
 \leq \rho^{A'}\mathds{1}^{BC}  - (r_1 ^2 - 1)E_3^{T_B} - E_1^{T_B}.
    \end{cases}\\
    \label{ppt-c-theorem}
    \textbf{PPT}\textbf{\textsubscript{C}}:&
    \begin{cases}
        \hspace{3.6cm} -(r_2 -1)E_2^{T_C} \leq E_1^{T_C} \leq (r_2 +1)E_2^{T_C},\\
-( \rho^{A'}\mathds{1}^{BC} - (r_2 ^2 - 1)E_2^{T_C} - E_1^{T_C} ) \leq   r_2(r_1 ^2-1) E_3^{T_C} \leq  \rho^{A'}\mathds{1}^{BC} - (r_2 ^2 - 1)E_2^{T_C} - E_1^{T_C}.
    \end{cases}
\end{align}

\begin{proof}
    We begin by obtaining the expression for the channel fidelity \eqref{channel-fidelity}. The desired expression follows from substituting the twirled Choi matrix \eqref{twirled-choi} into \eqref{channel-f-express}. Notably, only the first term, including $\phi^{A_1B'}   \phi^{A_2C'}$, contributes, while the other terms vanish due to orthogonality of the latter state to the following states: $
\phi^{A_1B'}   (\mathds{1}^{A_2C'}-\phi^{A_2C'}),(\mathds{1}^{A_1B'}-\phi^{A_1B'})   \phi^{A_2C'},
(\mathds{1}^{A_1B'}-\phi^{A_1B'})   (\mathds{1}^{A_2C'}-\phi^{A_2C'})$. 

 We next consider the constraints \eqref{CP}, \eqref{TP}, and \eqref{BCnotA}, which are necessary in order for the code to qualify as a CPTP map. Using Eqs. \eqref{twirled-choi} and \eqref{avg-twirled}, the $BC\not\to A$ condition, Eq. \eqref{BCnotA}, is equivalent to 
 \begin{align}
 \label{causality}
     & E_{1}^{A'BC}+ (r_2^2  - 1)E_{2}^{A'BC} + (r_1^2  - 1)E_{3}^{A'BC} + (r_2^2  - 1)(r_1^2  - 1)E_{4}^{A'BC}= \rho^{A'}\mathds{1}^{BC}.
 \end{align}
We will use this equality to eliminate $E_{4}^{A'BC}$ in the other constraints. By applying the trace-preserving condition, Eq. \ref{TP}, to Eq. \eqref{twirled-choi}, we obtain
\begin{align}
\label{tp-twiled}
    \qquad E_{1}^{BC}+ (r_2^2  - 1)E_{2}^{BC} + (r_1^2  - 1)E_{3}^{BC} + (r_2^2  - 1)(r_1^2  - 1)E_{4}^{BC}=\mathds{1}^{BC}.
\end{align}
Notice that in light of Eq. \eqref{causality}, the above trace-preserving condition is equivalent to the following condition
\begin{align}
\label{tp-rho}
\Tr\rho^{A'} = 1.
\end{align}
We can use either Eq. \eqref{tp-twiled} or \eqref{tp-rho} to enforce trace-preservation. We will see that in the no-signalling codes, in particular, when we impose $A\not\to BC$, using the condition \eqref{tp-twiled}, leads to the simple equality given by 
Eq. \eqref{tp+ns} which simultaneously satisfies both trace-preservation and no-signalling condition $A\not\to BC$ (as we will explain in detail shortly). For the codes without no-signalling constraint $A\not\to BC$, we use $\Tr\rho^{A'} = 1$ for simplicity. 

Since $\phi^{A_1B'}\phi^{A_2C'},\phi^{A_1B'}   (\mathds{1}^{A_2C'}-\phi^{A_2C'}),(\mathds{1}^{A_1B'}-\phi^{A_1B'})   \phi^{A_2C'}$, and $(\mathds{1}^{A_1B'}-\phi^{A_1B'})   (\mathds{1}^{A_2C'}-\phi^{A_2C'})$ project onto orthogonal subspaces and they are positive semidefinite operators, the twirled Choi matrix in Eq. \eqref{twirled-choi} satisfies the completely positive constraint \eqref{CP} if and only if $E_i^{A'BC}\geq 0$, for $i=1,2,3,4$. The first three of these are constraints in Eq. \eqref{cp-conditions}. Since we aim to eliminate $E_4^{A'BC}$ from all other constraints using \eqref{BCnotA}, we substitute the expression for $E_4^{A'BC}$ from \eqref{causality} into the constraint $E_4^{A'BC}\geq 0$, leading to Eq. \eqref{bcnottoa}. 

Now, if we want the forward-assisted code to be no-signalling from Alice to Bob and Charlie, i.e. satisfying Eq. \eqref{anottobc}, then this is equivalent to 
\begin{align*}
  & r_1 r_2  \big( \phi^{A_1B'}   \phi^{A_2C'} E_{1}^{ B C} +\phi^{A_1B'}       (\mathds{1} - \phi)^{A_2C'} E_{2}^{B C}
+(\mathds{1} - \phi)_{A_1B'} \phi_{A_2C'} E_{3}^{B C } + (\mathds{1} - \phi)_{A_1B'} (\mathds{1} - \phi)_{A_2C'} E_{4}^{B C }\big)\\
   =& \mathds{1}^{A_1A_2 B' C' }/r_1r_2 \otimes \big(E_{1}^{BC}+ (r_2^2   - 1)E_{2}^{BC} + (r_1^2  - 1)E_{3}^{BC} + (r_2^2  - 1)(r_1^2  - 1)E_{4}^{BC}\big).
\end{align*}

For the four orthogonal projectors on $A_1A_2B'C'$ on the left-hand-side to be sum to the identity, necessarily $E_1^{BC}=E_2^{BC}=E_3^{BC}=E_4^{BC}$. Using the trace-preserving constraint Eq. \eqref{tp-twiled}, the latter holds if and only if $E_i^{BC}=\mathds{1}^{BC}/r_1^2r_2^2$ for $i=1,2,3,4$. We eliminate $E_4^{BC}$ using Eq. \eqref{causality}, which comes with no additional constraint, and the remaining three constraints correspond to Eq. \eqref{tp+ns} in our theorem. 

Next, we consider the no-signalling constraints $B\not\to AC$ and $C\not\to AB$. As mentioned earlier, the three one-party-to-two-party no-signalling conditions ensure full no-signalling among the three parties; in particular, they imply $BC\not\to A$. However, we retain this latter constraint as it not only simplifies $B\not\to AC$ and $C\not\to AB$, as we will see shortly, but also simplifies other constraints more generally. We begin with $B\not\to AC$ and use the twirled Choi matrix \eqref{twirled-choi} in the constraint Eq. \eqref{bnottoac} as follows:
 \begin{align*}  
  \Tr_{B'} \bar{Z} &=
r_2 \mathds{1}^{A_1} \Big( \phi^{A_2C'}( E_{1}^{ A'B C} +(r_1^2-1)E_{3}^{A'B C }) +  
  (\mathds{1} - \phi)^{A_2C'} (E_{2}^{A'B C}
+(r_1^2-1)E_{4}^{A'B C })\Big)\\
&= \mathds{1}^{B}/\abs{B} \otimes r_2 \mathds{1}^{A_1}  \Big( \phi^{A_2C'}( E_{1}^{ A' C} +(r_1^2-1)E_{3}^{A' C } +  
  (\mathds{1} - \phi)^{A_2C'} (E_{2}^{A' C}
+(r_1^2-1)E_{4}^{A' C }\Big)\\
&=\mathds{1}^{B}/\abs{B} \otimes \Tr_{B'B}\bar{Z}^{ABCA'B'C'}
 \end{align*}
Noting that the projectors $\phi^{A_2C'}$ and $(\mathds{1} - \phi)^{A_2C'}$ span orthogonal subspaces, we obtain:
\begin{align*}
   & E_1^{A' BC} + (r_1^2 - 1)E_3^{A' BC} = \frac{\mathds{1}^B}{\abs{B}} \otimes \big(E_1^{A' C} + (r_1^2 - 1)E_3^{A' C}\big)\\
    & E_2^{A' BC} + (r_1^2 - 1)E_4^{A' BC} = \frac{\mathds{1}^B}{\abs{B}} \otimes \big(E_2^{A' C} + (r_1^2 - 1)E_4^{A' C}\big).
\end{align*}
The first equation above is constraint \eqref{BnototAC-one}, and the second equation is implied by combing the first and Eq. \eqref{causality}, so it is not included. Following a similar argument, the no-signalling condition $C\not\to AB$ is equivalent to 
\begin{align*}
    & E_1^{A' BC} + (r_2^2 - 1)E_2^{A' BC} = \frac{\mathds{1}^C}{\abs{C}} \otimes \big(E_1^{A' B} + (r_2^2 - 1)E_2^{A' B}\big)\\
    & E_3^{A' BC} + (r_2^2 - 1)E_4^{A' BC} = \frac{\mathds{1}^C}{\abs{C}} \otimes \big(E_3^{A' B} + (r_2^2 - 1)E_4^{A' B}\big).
\end{align*}
The first of these is constraint \eqref{CnototAB-one}, and the second equation is implied by combing the first and Eq. \eqref{causality}, so it is not included. 

Next, we consider PPT-preserving constraints. We begin with partial transpose on Alice's systems Eq. \eqref{ppt-a}. We will show that \eqref{twirled-choi} is PPT preserving on $AA'$ if and only if conditions \eqref{ppt-a-theorem} hold. To see this, note that as far as the PPT condition is concerned, a partial transpose on $AA'$ is equivalent to one on $BB'CC'$, since a standard transpose does not alter the eigenvalues of an operator. Moreover, we use the fact that $(\phi^{A_1B'})^{T_{A_1}} = (\mathbb{S}^{A_1B'} - \mathbb{A}^{A_1B'})/r_1$, where $\mathbb{S}^{A_1B'}$ and $\mathbb{A}^{A_1B'}$ are the projectors onto the symmetric and anti-symmetric subspaces of $A\otimes B'$, respectively. We also use that $\mathbb{S}^{A_1B'} + \mathbb{A}^{A_1B'} = \mathds{1}^{A_1B'}$ (and similarly for $\phi^{A_2C'}$). From these, we observe that 
\begin{align*}
    r_1(\phi^{A_1B'})^{T_{B'}} &= \mathbb{S}^{A_1B'} - \mathbb{A}^{A_1B'}\\
r_2(\phi^{A_2C'})^{T_{C'}} &= \mathbb{S}^{A_2C'} - \mathbb{A}^{A_2C'}\\
r_1(\mathds{1} - \phi_{B'A_1})^{T_{B'}} &= (r_1-1)\mathbb{S}^{A_1B'} + (r_1+1)\mathbb{A}^{A_1B'}\\
r_2(\mathds{1} - \phi^{A_2C'})^{T_{C'}} &= (r_2-1)\mathbb{S}^{A_2C'} + (r_2+1)\mathbb{A}^{A_2C'} .
\end{align*}
We now apply $T_{BB'CC'}$ to both sides of \eqref{twirled-choi} (Here and throughout the remainder of the proof, we omit the superscripts $A'BC$ from the operators $E_i$, $i=1,2,3,4$, for simplicity of notation). We obtain:
\begin{align*}    \widebar{Z}^{T_{B'BC'C}} =&  \mathbb{S}^{A_1B'}\mathbb{S}^{A_2C'}\left(E_1^{T_{BC}} + (r_2-1)E_2^{T_{BC}} + (r_1-1)E_3^{T_{BC}} + (r_2-1)(r_1-1)E_4^{T_{BC}}\right)\\
  &+ \mathbb{S}^{A_1B'}\mathbb{A}^{A_2C'}\left(-E_1^{T_{BC}} + (r_2+1)E_2^{T_{BC}} - (r_1-1)E_3^{T_{BC}} + (r_2+1)(r_1-1)E_4^{T_{BC}}\right)\\
  &+ \mathbb{A}^{A_1B'}\mathbb{S}^{A_2C'}\left(-E_1^{T_{BC}} - (r_2-1)E_2^{T_{BC}} + (r_1+1)E_3^{T_{BC}} + (r_2-1)(r_1+1)E_4^{T_{BC}}\right)\\
  &+ \mathbb{A}^{A_1B'}\mathbb{A}^{A_2C'}\left(E_1^{T_{BC}} - (r_2+1)E_2^{T_{BC}} - (r_1+1)E_3^{T_{BC}} + (r_2+1)(r_1+1)E_4^{T_{BC}}\right).
\end{align*}
Using the fact that each pair of symmetric/antisymmetric projectors is orthogonal to the others, we see that $\widebar{Z}^{T_{BB'CC'}} \geq 0$ is equivalent to requiring each individual line to be positive semi-definite. This, in turn, implies that the expressions inside the parentheses on each line must be positive semi-definite, i.e.
\begin{align*}
  &E_1^{T_{BC}} + (r_2 -1)E_2^{T_{BC}} + (r_1 -1)E_3^{T_{BC}} + (r_2 -1)(r_1 -1)E_4^{T_{BC}} \succeq 0,\\
  -&E_1^{T_{BC}} + (r_2 +1)E_2^{T_{BC}} - (r_1 -1)E_3^{T_{BC}} + (r_2 +1)(r_1 -1)E_4^{T_{BC}} \succeq 0,\\
  -&E_1^{T_{BC}} - (r_2 -1)E_2^{T_{BC}} + (r_1 +1)E_3^{T_{BC}} + (r_2 -1)(r_1 +1)E_4^{T_{BC}} \succeq 0,\\
  &E_1^{T_{BC}} - (r_2 +1)E_2^{T_{BC}} - (r_1 +1)E_3^{T_{BC}} + (r_2 +1)(r_1 +1)E_4^{T_{BC}} \succeq 0.
\end{align*}
We remove $E_4^{T_{BC}}$ using Eq. \eqref{causality} as follows: (it should now be clear why we chose to use the partial transpose over $BB'CC'$ rather than $AA'$, as the latter would result in a new operator $(\rho^{A'})^{T_{A}}$, whereas applying it to $\mathds{1}^{BC}$ has no effect)
\begin{align*}
(r_1r_2+r_1+r_2)E_1^{T_{BC}} + r_1(r_2^2-1)E_2^{T_{BC}}
  + r_2(r_1^2-1)E_3^{T_{BC}} +  \rho^{A'}\mathds{1}^{BC}\succeq 0\\
    -(r_1r_2-r_1+r_2)E_1^{T_{BC}} + r_1(r_2^2-1)E_2^{T_{BC}}
  - r_2(r_1^2-1)E_3^{T_{BC}} +  \rho^{A'}\mathds{1}^{BC}\succeq 0\\
    -(r_1r_2+r_1-r_2)E_1^{T_{BC}} - r_1(r_2^2-1)E_2^{T_{BC}}
  + r_2(r_1^2-1)E_3^{T_{BC}} +  \rho^{A'}\mathds{1}^{BC} \succeq 0\\
     (r_1r_2-r_1-r_2)E_1^{T_{BC}} - r_1(r_2^2-1)E_2^{T_{BC}}
  - r_2(r_1^2-1)E_3^{T_{BC}} +  \rho^{A'}\mathds{1}^{BC} \succeq 0,
\end{align*}
These in turn can be equivalently but more compactly written as follows:
\begin{align*}
  - \frac{1}{r_1r_2}\rho^{A'}\mathds{1}^{BC}  - (1+\frac{1}{r_2}+\frac{1}{r_1})E_1^{T_{BC}}\leq \frac{(r_2^2-1)}{r_2}E_2^{T_{BC}}
  + \frac{(r_1^2-1)}{r_1}E_3^{T_{BC}} \leq \frac{1}{r_1r_2} \rho^{A'}\mathds{1}^{BC} + (1-\frac{1}{r_2}-\frac{1}{r_1})E_1^{T_{BC}},\\
 - \frac{1}{r_1r_2}\rho^{A'}\mathds{1}^{BC} +  (1-\frac{1}{r_2}+\frac{1}{r_1})E_1^{T_{BC}}\leq   \frac{(r_2^2-1)}{r_2}E_2^{T_{BC}}
  - \frac{(r_1^2-1)}{r_1}E_3^{T_{BC}} \leq \frac{1}{r_1r_2} \rho^{A'}\mathds{1}^{BC} - (1+\frac{1}{r_2}-\frac{1}{r_1})E_1^{T_{BC}}.
\end{align*}
These are the constraints Eq. \eqref{ppt-a-theorem} in our theorem.

We continue by applying the partial transpose on Bob's systems as in Eq. \eqref{ppt-b}.
\begin{align*}
  \widebar{Z}^{T_{BB'}} =& r_2 (\mathbb{S}^{A_1B'} - \mathbb{A}^{A_1B'}) \phi^{A_2C'} E_1^{T_B} + r_2 (\mathbb{S}^{A_1B'} - \mathbb{A}^{A_1B'})(\mathds{1} - \phi^{A_2C'})E_2^{T_B}\\
                    &+ r_2 \big((r_1 -1)\mathbb{S}^{A_1B'} + (r_1 +1)\mathbb{A}^{A_1B'}\big)\phi^{A_2C'}E_3^{T_B}\\
                    &+ r_2 \big((r_1 -1)\mathbb{S}^{A_1B'} + (r_1 +1)\mathbb{A}^{A_1B'}\big)(\mathds{1} - \phi^{A_2C'})E_4^{T_B}\\
  =& r_2 \phi^{A_2C'}\left((\mathbb{S}^{A_1B'}-\mathbb{A}^{A_1B'})E_1^{T_B} + \big((r_1 -1)\mathbb{S}^{A_1B'} + (r_1 +1)\mathbb{A}^{A_1B'}\big)E_3^{T_B}\right)\\
  &+ r_2 (\mathds{1}-\phi^{A_2C'})\left((\mathbb{S}^{A_1B'}-\mathbb{A}^{A_1B'})E_2^{T_B} + \big((r_1 -1)\mathbb{S}^{A_1B'} + (r_1 +1)\mathbb{A}^{A_1B'}\big)E_4^{T_B}\right)
\end{align*}
Since \(\phi^{A_2C'}\) and \(\mathds{1}-\phi^{A_2C'}\) are orthogonal projectors, one has $\widebar{Z}^{T_{BB'}}\geq 0$ if and only if
\begin{align*}
&(\mathbb{S}^{A_1B'}-\mathbb{A}^{A_1B'})E_1^{T_B} + ((r_1 -1)\mathbb{S}^{A_1B'} + (r_1 +1)\mathbb{A}^{A_1B'})E_3^{T_B} \geq 0 ,\\
&(\mathbb{S}^{A_1B'}-\mathbb{A}^{A_1B'})E_2^{T_B} + ((r_1 -1)\mathbb{S}^{A_1B'} + (r_1 +1)\mathbb{A}^{A_1B'})E_4^{T_B} \geq 0.
\end{align*}
Grouping each line according to $\mathbb{S}^{A_1B'}$ and $\mathbb{A}^{A_1B'}$, and using the orthogonality of these projectors, we find that each line, respectively, is positive semidefinite if and only if the following conditions hold:  
\begin{align*}
    -(r_1 -1)E_3^{T_B} \leq E_1^{T_B} \leq (r_1 +1)E_3^{T_B},\\
    -(r_1 -1)E_4^{T_B} \leq E_2^{T_B} \leq (r_1 +1)E_4^{T_B}.
\end{align*}
We use Eq. \eqref{causality} to omit $E_4^{T_B}$ in the second line as follows:
\begin{align*}
- (   \rho^{A'}\mathds{1}^{BC} - (r_1 ^2 - 1)E_3^{T_B} - E_1^{T_B})    \leq r_1(r_2 ^2-1) E_2^{T_B}
 \leq  \rho^{A'}\mathds{1}^{BC}  - (r_1 ^2 - 1)E_3^{T_B} - E_1^{T_B}.
\end{align*}
The two conditions are Eq. \eqref{ppt-b-theorem} in our theorem. The partial transpose on Chalie's systems $CC'$ follows along the same lines, so that $\widebar{Z}^{T_{CC'}}\geq 0$ if and only if 
\begin{align*}
    -(r_2 -1)E_2^{T_C} \leq E_1^{T_C} \leq (r_2 +1)E_2^{T_C}, \\
     -(r_2 -1)E_4^{T_C} \leq E_3^{T_C} \leq (r_2 +1)E_4^{T_C}.
\end{align*}
Removing $E_4^{T_C}$ in the second line using Eq. \eqref{causality}, we obtain
\begin{align*}
 -( \rho^{A'}\mathds{1}^{BC} - (r_2 ^2 - 1)E_2^{T_C} - E_1^{T_C} ) \leq   r_2(r_1 ^2-1) E_3^{T_C} \leq  \rho^{A'}\mathds{1}^{BC} - (r_2 ^2 - 1)E_2^{T_C} - E_1^{T_C}.
\end{align*}
These are constraints Eq. \eqref{ppt-c-theorem} in our theorem. This concludes the proof.
\end{proof}

\bigskip

\begin{corollary}
   The optimal solution of the SDP in Theorem \ref{full-fledged} corresponds to $F^{\textbf{NS}\cap\textbf{PPT}}(\mathcal{N},r_1,r_2)$. To obtain $F^{\textbf{NS}}(\mathcal{N},r_1,r_2)$, we maximize the expression in Eq. \eqref{channel-fidelity} subject only to the no-signalling conditions Eqs. \eqref{tp+ns}, \eqref{BnototAC-one}, and \eqref{CnototAB-one}. Similarly, to obtain $F^{\textbf{PPT}}(\mathcal{N},r_1,r_2)$, we maximize Eq. \eqref{channel-fidelity} subject only to the PPT-preserving conditions Eqs. \eqref{ppt-a-theorem}, \eqref{ppt-b-theorem}, and \eqref{ppt-c-theorem}, and include $\Tr\rho^{A'}=1$ (the latter is required when \eqref{anottobc} is dropped, as it implicitly enforces trace-preserving condition, see the proof of Theorem \ref{full-fledged} for details).   
\end{corollary}

\bigskip
In the following we show that the SDP corresponding to $F^{\textbf{NS}}(\mathcal{N},r_1,r_2)$ can be significantly simplified.
\begin{proposition}
\label{no-signalling-fidelity}
For a quantum broadcast channel $\mathcal{N}^{A'\to BC}$, there is a forward-assisted code of size $(r_1,r_2)$, average channel input $\rho^{A'}$ and channel fidelity $f_c$, which is no-signalling if and only if the following SDP has a feasible solution
    \begin{align*}
             f_c &=
    \Tr \Lambda^{A'BC}(N^{A'BC})^{T}\\
   0 &\leq \Lambda^{A'BC} \leq \rho^{A'}\mathds{1}^{BC}, \,  \Lambda^{BC}=\frac{\mathds{1}^{BC}}{r_1^2r_2^2}
\end{align*}
\end{proposition}
\begin{proof}
   This follows by applying the Fourier-Motzkin elimination to the no-signaling SDP (including Eqs. \eqref{channel-fidelity} to \eqref{CnototAB-one}) in Theorem \ref{full-fledged}, to eliminate the variables $E_2^{A'BC}$ and $E_3^{A'BC}$.
\end{proof}

The key insight from this result is that no-signalling codes with the only no-signalling condition being from the sender to the receivers, i.e., $A\not\to BC$, perform no better than a no-signalling code with full no-signalling among the three parties. However, when the code is further constrained to be PPT-preserving, these no-signalling conditions may limit its effectiveness. Thus, without PPT-preserving restrictions, fully no-signalling codes are as effective as no-signalling codes with $A\not\to BC$. The reason is that the Fourier-Motzkin elimination does not generally work after introducing PPT-preserving conditions.

\medskip

The exact SDP presented in Theorem \ref{full-fledged} is too complex to use for finding bounds on EPR pair generation in quantum networks. In particular, when the code size $(r_1, r_2)$ becomes a variable, Theorem \ref{full-fledged} leads to a highly nonlinear optimization problem, which may even resist relaxations aimed at maximizing the code size. For this reason, the following proposition introduces a relaxation of Theorem \ref{full-fledged}, which proves particularly useful for deriving converse bounds. Indeed, the specific relaxations chosen allow us to establish converse bounds with desirable properties.

\begin{proposition}
\label{fidelity-relax-1}
    For any quantum broadcast channel $\mathcal{N}^{A'\to BC}$, its channel fidelity, assisted with PPT-preserving and no-signalling codes, is bounded as follows: 
\begin{align*}
  F_{\text{PPT} \,\cap\, \text{NS}}(\mathcal{N},r_1,r_2) & \,\leq \, \max \, \Tr \Lambda^{A'BC}(N^{A'BC})^T\\
    & \text{s.t.}   \hspace{0.4cm} 0 \leq \Lambda^{A'BC}\leq  \rho^{A'}\mathds{1}^{BC},\,\,\,\Tr\rho^{A'}=1, \,\,\, \Lambda^{BC}=\frac{1}{r_1^2r_2^2}\mathds{1}^{BC},\\
    &\hspace{0.8cm} - \frac{1}{r_1r_2}\rho^{A'}\mathds{1}^{BC} \leq  (\Lambda^{A'BC})^{T_{BC}} \leq \frac{1}{r_1r_2}\rho^{A'}\mathds{1}^{BC},\\
   & \hspace{0.8cm} - \rho^{A'}\mathds{1}^{BC} \leq  (\Lambda^{A'BC})^{T_{B}} \leq \rho^{A'}\mathds{1}^{BC},\\
      & \hspace{0.8cm} - \rho^{A'}\mathds{1}^{BC} \leq  (\Lambda^{A'BC})^{T_{C}} \leq \rho^{A'}\mathds{1}^{BC}.
\end{align*}
\end{proposition}
\begin{proof}
The relaxation involves eliminating all variables except $E_1^{A'BC}$, denoted as $\Lambda^{A'BC}$, and $\rho^{A'}$. More specifically, for the PPT conditions, we determine the implications of the inequalities on $E_1^{A'BC}$, while for the remaining conditions, we apply Fourier-Motzkin elimination.
\end{proof}

We are now ready to derive the one-shot characterization of quantum communication over quantum broadcast channels assisted by NS or PPT-preserving codes. This derivation will be presented in the next section.

\medskip

\section{Converse bounds for non-asymptotic EPR pair generation}
\label{weak-converse}
In the following, we explicitly formulate the optimization problem to determine the $\varepsilon$-one-shot sum-capacity of a quantum broadcast channel assisted by PPT-preserving and/or no-signalling codes. This problem is highly nonlinear and therefore not an SDP. We then explore various SDP relaxations to approximate its solution. The following proposition is based on Theorem \ref{full-fledged}. Alternatively, one can use the corresponding constraints in Proposition \ref{full-fledged-fourier-motzkin} equivalently.

\begin{proposition}
\label{full-fledged-os-capacity}
For any quantum broadcast channel $\mathcal{N}^{A'\to BC}$ with Choi matrix $N^{A'BC}$ and a specified error $\varepsilon$, the $\varepsilon$-one-shot sum-capacity of the channel, assisted by PPT-preserving and non-signaling codes, can be expressed as the solution of following optimization problem (using the notation of Theorem \ref{full-fledged}):
\begin{align*}
    Q^{\textbf{PPT} \,\cap\, \textbf{NS}}_{1}(\mathcal{N},\varepsilon) = \max &\log (r_1r_2),\\
    &\text{s.t.}\,\,\,\, f_c =
    \Tr E_1^{A'BC}(N^{A'BC})^T\geq 1 -\varepsilon,\\
& \,\,\,\,\,\,\,\,\,\,\,\, \text{and constraints}\,\, \eqref{cp-conditions} \,\,to\,\, \eqref{ppt-c-theorem}\,\, \text{in Theorem}\,\, \ref{full-fledged}.
\end{align*}
To obtain $Q^{\textbf{PPT}}_{1}(\mathcal{N},\varepsilon)$ or $Q^{\textbf{NS}}_{1}(\mathcal{N},\varepsilon)$, we maximize $\log(r_1r_2)$ subject to the appropriate constraints: For $Q^{\textbf{NS}}_{1}(\mathcal{N},\varepsilon)$, the constraints are Eqs. \eqref{cp-conditions} to \eqref{CnototAB-one}. For $Q^{\textbf{PPT}}_{1}(\mathcal{N},\varepsilon)$, the constraints are Eqs. \eqref{ppt-a-theorem} to \eqref{ppt-c-theorem}, and we also include $\Tr\rho^{A'}=1$. In both cases, we also impose the conditions \eqref{cp-conditions} and \eqref{bcnottoa}.
\end{proposition}

In the following, we present two different relaxations of the above proposition. While these problems remain nonlinear and are therefore not SDPs, they have a transparent formulation that allows for further relaxation into SDPs. The first one below is based on Proposition \ref{fidelity-relax-1}.

\begin{proposition}
\label{non-sdp-relax-1}
    For any quantum broadcast channel $\mathcal{N}^{A'\to BC}$, its $\varepsilon$-one-shot sum-capacity, assisted with PPT-preserving and no-signalling codes, is bounded as $Q^{\textbf{PPT} \,\cap\, \textbf{NS}}_{1}(\mathcal{N},\varepsilon)\leq -\log g(\mathcal{N},\varepsilon)$, where 
\begin{align*}
    g(\mathcal{N},\varepsilon)& = \min \Tr S^{A'}\\
    &\text{s.t.}\,\,\,\,\,\, f_c =
    \Tr \Lambda^{A'BC}(N^{A'BC})^T \geq 1-\varepsilon,\\
     &  \hspace{0.9cm} 0 \leq \Lambda^{A'BC}\leq  \rho^{A'}\mathds{1}^{BC},\,\,\,\Tr\rho^{A'}=1, \Lambda^{BC}=m^2\mathds{1}^{BC},\\
    &\hspace{0.7cm} - S^{A'}\mathds{1}^{BC} \leq  (\Lambda^{A'BC})^{T_{BC}} \leq S^{A'}\mathds{1}^{BC},\\
   & \hspace{0.7cm} - \rho^{A'}\mathds{1}^{BC} \leq  (\Lambda^{A'BC})^{T_{B}} \leq \rho^{A'}\mathds{1}^{BC},\\
      & \hspace{0.7cm} - \rho^{A'}\mathds{1}^{BC} \leq  (\Lambda^{A'BC})^{T_{C}} \leq \rho^{A'}\mathds{1}^{BC}.
\end{align*}
\end{proposition}
\begin{proof}
Consider the following nonlinear optimization program inspired by Proposition \ref{fidelity-relax-1}:
\begin{align*}
     \max & \,\,\, -\log \frac{1}{r_1r_2} \\
    &\text{s.t.}\,\,\,\,
    \Tr \Lambda^{A'BC}(N^{A'BC})^T \geq 1-\varepsilon,\\
     &  \hspace{0.9cm} 0 \leq \Lambda^{A'BC}\leq  \rho^{A'}\mathds{1}^{BC},\,\,\,\Tr\rho^{A'}=1, \Lambda^{BC}=\frac{1}{r_1^2r_2^2}\mathds{1}^{BC},\\
    &\hspace{0.7cm} - \frac{1}{r_1r_2} \rho^{A'}\mathds{1}^{BC} \leq  (\Lambda^{A'BC})^{T_{BC}} \leq \frac{1}{r_1r_2}\rho^{A'}\mathds{1}^{BC},\\
   & \hspace{0.7cm} - \rho^{A'}\mathds{1}^{BC} \leq  (\Lambda^{A'BC})^{T_{B}} \leq \rho^{A'}\mathds{1}^{BC},\\
      & \hspace{0.7cm} - \rho^{A'}\mathds{1}^{BC} \leq  (\Lambda^{A'BC})^{T_{C}} \leq \rho^{A'}\mathds{1}^{BC}.
\end{align*}
The first requirement follows from the definition of $\varepsilon$-one-shot sum-capacity. 
We define $m=\frac{1}{r_1r_2}$ and $S^{A'} = m\rho^{A'}$. Hence, $\Tr S^{A'} = m$.
\end{proof}

\bigskip

We present another relaxation based on Theorem \ref{full-fledged-fourier-motzkin}. Although this relaxation remains nonlinear again, it provides valuable insights for the subsequent steps towards a linear program.

\begin{proposition}
\label{non-sdp-relax-2}
    For any quantum broadcast channel $\mathcal{N}^{A'\to BC}$, its $\varepsilon$-one-shot sum-capacity, assisted with PPT-preserving and no-signalling codes, is bounded as $Q^{\textbf{PPT} \,\cap\, \textbf{NS}}_{1}(\mathcal{N},\varepsilon)\leq -\log h(\mathcal{N},\varepsilon)$, where 
\begin{align*}
 h(\mathcal{N},\varepsilon)& = \min \Tr S^{A'}\\
    &\text{s.t.}\,\,\,\,\,\, f_c =
    \Tr E_1^{A'BC}(N^{A'BC})^T \geq 1-\varepsilon,\\
     &\hspace{0.9cm} 0 \leq E^{A'BC}_1\leq \big(E_1^{A' C} + S_3^{A'C}\big) \otimes \frac{\mathds{1}^B}{\abs{B}}, \big(E_1^{A' B} + S_2^{A' B}\big) \otimes \frac{\mathds{1}^C}{\abs{C}},\\
    &\hspace{0.9cm} \big(E_1^{A' B} + S_2^{A' B}\big) \otimes \frac{\mathds{1}^C}{\abs{C}} + \big(E_1^{A' C} + S_3^{A' C}\big) \otimes \frac{\mathds{1}^B}{\abs{B}} \leq  \rho^{A'}\mathds{1}^{BC} + E_1^{A'BC}.\\
    & \hspace{0.9cm}\Tr\rho^{A'}=1, E_1^{BC}=m^2\mathds{1}^{BC},S_3^{C}=m^2\abs{B}\mathds{1}^{C},
     S_2^{B}=m^2\abs{C}\mathds{1}^{B},\\
    &\hspace{0.8cm}- S^{A'}\mathds{1}^{BC} \leq  (E_1^{A'BC})^{T_{BC}} \leq S^{A'}\mathds{1}^{BC},\\
   & \hspace{0.8cm}- \rho^{A'}\mathds{1}^{BC} \leq  (E_1^{A'BC})^{T_{B}} \leq \rho^{A'}\mathds{1}^{BC},\\
      &  \hspace{0.8cm}- \rho^{A'}\mathds{1}^{BC} \leq  (E_1^{A'BC})^{T_{C}} \leq \rho^{A'}\mathds{1}^{BC}.
\end{align*}
\end{proposition}
\begin{proof}
    This relaxation is obtained using Proposition \ref{full-fledged-fourier-motzkin}. The first line follows from the definition of $\varepsilon$-one-shot sum-capacity. We let $m=1/r_1r_2$ and define operators $S^{A'} = m\rho^{A'}$, $S_2^{A'B}=(r_2^2-1)E_{2}^{A'B}$, and $S_3^{A'B}=(r_1^2-1)E_{3}^{A'C}$. Then the first three lines are derived according to the first three constraints in Proposition \ref{full-fledged-fourier-motzkin}. The last three lines are derived similar to last three lines of Proposition \ref{non-sdp-relax-1}.
\end{proof}

The optimization problems in Propositions \ref{non-sdp-relax-1} and \ref{non-sdp-relax-2} are not SDPs because of the nonlinear term $m^2$ emerging from the no-signalling conditions. In the following, we present a hierarchy of SDP relaxations for each optimization problem. Specifically, we propose two distinct approaches to handle the nonlinear variable $m^2$ and derive the relaxations, inspired by \cite{wang-q}. The first approach consists in replacing $m^2$ with a variable $0 \leq t \leq 1$. 
Based on the nonlinear programs in Propositions \ref{non-sdp-relax-1} and \ref{non-sdp-relax-2}, we denote the SDPs resulting from this relaxation method as $\widetilde{g}(\mathcal{N},\varepsilon)$ and $\widetilde{h}(\mathcal{N},\varepsilon)$ corresponding to each proposition, respectively.

The second approach involves introducing an initial estimate for the value of $m$, possibly derived from Proposition \ref{full-fledged-os-capacity}, denoted by $\widehat{m}$, which satisfies the inequality
\begin{align*}
  Q^{\textbf{PPT} \cap \textbf{NS}}_{1}(\mathcal{N},\varepsilon) \leq -\log \widehat{m}.   
\end{align*}  
This initial value, $\widehat{m}$, is incorporated into the optimization problems from Propositions \ref{non-sdp-relax-1} and \ref{non-sdp-relax-2} by introducing a variable $t$ that searches in $\widehat{m}^2 \leq t < 1$ and improves upon the initial estimate of $\widehat{m}$. The value of $\widehat{m}$ is then refined iteratively. This results in the new relaxations $\widehat{g}(\mathcal{N},\varepsilon)$ and $\widehat{h}(\mathcal{N},\varepsilon)$, respectively. Although this method can achieve tighter bounds, it requires an initial step to find an estimate for $\widehat{m}$, and its refinement is gradual, necessitating repeated iterations of the program to obtain improved results.

Before presenting the relation between these upper bounds, we express them more formally as follows in the order they were explained:
\begin{align*}
    \widetilde{g}(\mathcal{N},\varepsilon)& = \min \Tr S^{A'}\\
    &\text{s.t.}\,\,\,\,\,\, f_c =
    \Tr E_1^{A'BC}(N^{A'BC})^{T} \geq 1-\varepsilon,\\
     &  \hspace{0.9cm} 0 \leq E^{A'BC}_1\leq  \rho^{A'}\mathds{1}^{BC},\,\,\,\Tr\rho^{A'}=1, E_1^{BC}=t\mathds{1}^{BC},\\
    &\hspace{0.7cm} - S^{A'}\mathds{1}^{BC} \leq  (E_1^{A'BC})^{T_{BC}} \leq S^{A'}\mathds{1}^{BC},\\
   & \hspace{0.7cm} - \rho^{A'}\mathds{1}^{BC} \leq  (E_1^{A'BC})^{T_{B}} \leq \rho^{A'}\mathds{1}^{BC},\\
      & \hspace{0.7cm} - \rho^{A'}\mathds{1}^{BC} \leq  (E_1^{A'BC})^{T_{C}} \leq \rho^{A'}\mathds{1}^{BC}.
\\
\\
 \widetilde{h}(\mathcal{N},\varepsilon)& = \min \Tr S^{A'}\\
    &\text{s.t.}\,\,\,\,\,\, f_c =
    \Tr E_1^{A'BC}(N^{A'BC})^{T} \geq 1-\varepsilon,\\
     &\hspace{0.9cm} 0 \leq E^{A'BC}_1\leq \big(E_1^{A' C} + S_3^{A'C}\big) \otimes \frac{\mathds{1}^B}{\abs{B}}, \big(E_1^{A' B} + S_2^{A' B}\big) \otimes \frac{\mathds{1}^C}{\abs{C}},\\
    &\hspace{0.9cm} \big(E_1^{A' B} + S_2^{A' B}\big) \otimes \frac{\mathds{1}^C}{\abs{C}} + \big(E_1^{A' C} + S_3^{A' C}\big) \otimes \frac{\mathds{1}^B}{\abs{B}} \leq  \rho^{A'}\mathds{1}^{BC} + E_1^{A'BC}.\\
    & \hspace{0.9cm}\Tr\rho^{A'}=1, E_1^{BC}=t\mathds{1}^{BC},S_3^{C}=t\abs{B}\mathds{1}^{C},
     S_2^{B}=t\abs{C}\mathds{1}^{B},\\
    &\hspace{0.8cm}- S^{A'}\mathds{1}^{BC} \leq  (E_1^{A'BC})^{T_{BC}} \leq S^{A'}\mathds{1}^{BC},\\
   & \hspace{0.8cm}- \rho^{A'}\mathds{1}^{BC} \leq  (E_1^{A'BC})^{T_{B}} \leq \rho^{A'}\mathds{1}^{BC},\\
      &  \hspace{0.8cm}- \rho^{A'}\mathds{1}^{BC} \leq  (E_1^{A'BC})^{T_{C}} \leq \rho^{A'}\mathds{1}^{BC}.
\\
\\
    \widehat{g}(\mathcal{N},\varepsilon)& = \min \Tr S^{A'}\\
    &\text{s.t.}\,\,\,\,\,\, f_c =
    \Tr E_1^{A'BC}(N^{A'BC})^{T} \geq 1-\varepsilon,\\
     &  \hspace{0.9cm} 0 \leq E^{A'BC}_1\leq  \rho^{A'}\mathds{1}^{BC},\,\,\,\Tr\rho^{A'}=1, E_1^{BC}=t\mathds{1}^{BC}, t\geq \widehat{m}^2,\\
    &\hspace{0.7cm} - S^{A'}\mathds{1}^{BC} \leq  (E_1^{A'BC})^{T_{BC}} \leq S^{A'}\mathds{1}^{BC},\\
   & \hspace{0.7cm} - \rho^{A'}\mathds{1}^{BC} \leq  (E_1^{A'BC})^{T_{B}} \leq \rho^{A'}\mathds{1}^{BC},\\
      & \hspace{0.7cm} - \rho^{A'}\mathds{1}^{BC} \leq  (E_1^{A'BC})^{T_{C}} \leq \rho^{A'}\mathds{1}^{BC}.
      \end{align*}
\begin{align*}
 \widehat{h}(\mathcal{N},\varepsilon)& = \min \Tr S^{A'}\\
    &\text{s.t.}\,\,\,\,\,\, f_c =
    \Tr E_1^{A'BC}(N^{A'BC})^{T} \geq 1-\varepsilon,\\
     &\hspace{0.9cm} 0 \leq E^{A'BC}_1\leq \big(E_1^{A' C} + S_3^{A'C}\big) \otimes \frac{\mathds{1}^B}{\abs{B}}, \big(E_1^{A' B} + S_2^{A' B}\big) \otimes \frac{\mathds{1}^C}{\abs{C}},\\
    &\hspace{0.9cm} \big(E_1^{A' B} + S_2^{A' B}\big) \otimes \frac{\mathds{1}^C}{\abs{C}} + \big(E_1^{A' C} + S_3^{A' C}\big) \otimes \frac{\mathds{1}^B}{\abs{B}} \leq  \rho^{A'}\mathds{1}^{BC} + E_1^{A'BC}.\\
    & \hspace{0.9cm}\Tr\rho^{A'}=1, E_1^{BC}=t\mathds{1}^{BC},S_3^{C}=t\abs{B}\mathds{1}^{C},
     S_2^{B}=t\abs{C}\mathds{1}^{B},   t \geq \widehat{m}^2 ,\\
    &\hspace{0.8cm}- S^{A'}\mathds{1}^{BC} \leq  (E_1^{A'BC})^{T_{BC}} \leq S^{A'}\mathds{1}^{BC},\\
   & \hspace{0.8cm}- \rho^{A'}\mathds{1}^{BC} \leq  (E_1^{A'BC})^{T_{B}} \leq \rho^{A'}\mathds{1}^{BC},\\
      &  \hspace{0.8cm}- \rho^{A'}\mathds{1}^{BC} \leq  (E_1^{A'BC})^{T_{C}} \leq \rho^{A'}\mathds{1}^{BC}.
\end{align*}

We present the comparison of the bounds derived so far in the following theorem.
\begin{theorem}
\label{comparison}
    For any quantum broadcast channel $\mathcal{N}^{A'\to BC}$, we obtain the following chain of inequalities among the upper bounds on its $\varepsilon$-one-shot sum-capacity, assisted with PPT-preserving and no-signalling codes:
    \begin{align*}
      Q^{(1)}(\mathcal{N},\varepsilon)&\leq  
      Q^{\textbf{PPT} \,\cap\, \textbf{NS}}_{1}(\mathcal{N},\varepsilon) \\
      &\leq -\log \widehat{h}(\mathcal{N},\varepsilon) \\ 
      & \leq -\log \widetilde{h}(\mathcal{N},\varepsilon)\\
     & \leq -\log \widetilde{g}(\mathcal{N},\varepsilon).
    \end{align*}
    Also
     \begin{align*}
      Q^{(1)}(\mathcal{N},\varepsilon)&\leq  
      Q^{\textbf{PPT} \,\cap\, \textbf{NS}}_{1}(\mathcal{N},\varepsilon) \\
      &\leq -\log \widehat{h}(\mathcal{N},\varepsilon)\\
      &\leq -\log \widehat{g}(\mathcal{N},\varepsilon) \\ 
     & \leq -\log \widetilde{g}(\mathcal{N},\varepsilon),
    \end{align*}
    Note that the two sets of inequalities only differ in their last two lines.
\end{theorem}
\begin{proof}
    In the first set of inequalities, the first inequality is trivial by definition. The second inequality follows because the optimal solution for $Q^{\textbf{PPT} \,\cap\, \textbf{NS}}_{1}(\mathcal{N},\varepsilon)$ is a feasible point for $\widehat{h}(\mathcal{N},\varepsilon)$. Specifically, let the optimal solution of $Q^{\textbf{PPT} \,\cap\, \textbf{NS}}_{1}(\mathcal{N},\varepsilon)$ be given by  
\[
(E_1^{A'BC}, E_2^{A'B}, E_3^{A'C}, \rho^{A'}, m = \frac{1}{r_1r_2}).
\]  
(Remember that this SDP arises as a relaxation of Proposition \ref{full-fledged-fourier-motzkin}.) Define  
\begin{align*}
  S_2^{A'B} = (r_2^2-1)E_2^{A'B}, \quad S_3^{A'C} = (r_1^2-1)E_3^{A'C}, \quad S^{A'} = m\rho^{A'}.  
\end{align*}  
It is easy to verify that $\{S_2^{A'B}, S_3^{A'C}, S^{A'}, \rho^{A'}, m\}$ forms a feasible solution to the SDP defining $\widehat{h}(\mathcal{N},\varepsilon)$. Therefore,
\begin{align*}
  \widehat{h}(\mathcal{N},\varepsilon) \leq \Tr S^{A'} = m,  
\end{align*} 
which implies 
\begin{align*}
 Q^{\textbf{PPT} \,\cap\, \textbf{NS}}_{1}(\mathcal{N},\varepsilon) = -\log m \leq -\log \widehat{h}(\mathcal{N},\varepsilon).   
\end{align*}
The third inequality follows easily, as the minimization is performed over a smaller feasible set, which results in a larger value, leading to a higher rate or a looser upper bound.  

For the second set of inequalities, the first and second inequalities are identical to those in the first set. The third and fourth inequalities hold for the same reason: the minimization is carried out over a smaller feasible set.
\end{proof}

\begin{remark}
    Concerning Theorem \ref{comparison}, the only case that is not immediately clear is whether $\widetilde{h}(\mathcal{N},\varepsilon) \lessgtr \widehat{g}(\mathcal{N},\varepsilon)$. On one hand, $\widetilde{h}(\mathcal{N},\varepsilon)$ is obtained by minimizing over a smaller set of feasible points. On the other hand, $\widehat{g}(\mathcal{N},\varepsilon)$ refines its value iteratively. Consequently, the final result may depend on the specific properties of the channel.
\end{remark}
We now turn to strong converse bounds. As we will see shortly, obtaining strong converses requires additional relaxations. This is expected, as strong converses impose significantly stronger constraints on the communication rates.

\bigskip

\section{Strong converse bound for EPR pair generation over quantum broadcast channels}
\label{sec:strong-converse}
We begin with the following relaxation as an upper bound on the $\varepsilon$-one-shot sum-capacity $Q^{\textbf{PPT} \cap \textbf{NS}}_{1}(\mathcal{N},\varepsilon)$. This relaxation builds on the SDP formulation introduced in Proposition \ref{non-sdp-relax-1} and leads to a strong converse similar to the one in \cite{wang-q} for point-to-point channels. However, our bound is tighter, as it incorporates a specific no-signaling condition while preserving the required properties.

\begin{proposition}
    \label{non-sdp-relax-for-strong-converse}
    For any quantum broadcast channel $\mathcal{N}^{A'\to BC}$, the $\varepsilon$-one-shot sum-capacity, assisted by PPT-preserving and no-signalling codes, is bounded as follows:
    \begin{align*}
    Q^{\textbf{PPT} \,\cap\, \textbf{NS}}_{1}(\mathcal{N},\varepsilon)&\leq - \log \min m \\
    &\text{s.t.}\,\,\,\,\,\,
    \Tr \Lambda^{A'BC}(N^{A'BC})^T \geq 1-\varepsilon,\\
     &  \hspace{0.9cm}  \Lambda^{A'BC},\, \rho^{A'} \geq 0,\,\,\,\Tr\rho^{A'}=1, \,\,\,\Lambda^{BC}=m^2\mathds{1}^{BC},\\
    &\hspace{0.7cm} - m \rho^{A'}\mathds{1}^{BC} \leq  (\Lambda^{A'BC})^{T_{BC}} \leq m\rho^{A'}\mathds{1}^{BC}.
\end{align*}
\end{proposition}
\begin{proof}
This formulation is similar to the SDP in Proposition \ref{non-sdp-relax-1}, but with certain constraints removed.
\end{proof}

\bigskip
\bigskip
We shall then continue as follows:

\begin{proposition}
\label{strong-1}
    For any quantum broadcast channel $\mathcal{N}^{A'\to BC}$ and a specified error $\varepsilon$, the $\varepsilon$-one-shot sum-capacity of the channel, assisted by PPT-preserving and no-signalling codes, is bounded as follows:
\begin{align}
\label{one-shot-converse}
    Q^{\textbf{PPT} \,\cap\, \textbf{NS}}_{1}(\mathcal{N},\varepsilon) \leq Q_{\Gamma}(\mathcal{N}) - \log(1-\varepsilon),
    \end{align}
    where $Q_{\Gamma}(\mathcal{N}) = \log\Gamma(\mathcal{N})$, and
    \begin{align}
    \nonumber
        \text{(Primal)}\,\,\,\,\,\,\,\,\,  \Gamma(\mathcal{N})  = &\max \, \Tr R^{A'BC}(N^{A'BC})^T\\ \label{primal}
    & \text{s.t.}
       \hspace{0.4cm} \, R^{A'BC},\rho^{A'}\geq 0,\,\,\,\Tr\rho^{A'}=1,\,\,\, R^{BC} = m\mathds{1}^{BC},\\
       \nonumber
    &\hspace{0.6cm} \,\,\,- \rho^{A'}\mathds{1}^{BC} \leq  (R^{A'BC})^{T_{BC}} \leq \rho^{A'}\mathds{1}^{BC},
    \end{align}
\begin{align}
\nonumber
        \text{(Dual)}\,\,\,\,\,\,\,\,\,  \Gamma(\mathcal{N}) & =  \min \, \mu\\
        \nonumber
    & \text{s.t.}
       \hspace{0.3cm} (N^{A'BC})^T \leq (V^{A'BC} - Y^{A'BC})^{T_{BC}} + \mathds{1}^{A'}W^{BC},\\ \label{dual}
    &\hspace{0.7cm} \Tr_{BC}\{V^{A'BC} + Y^{A'BC}\}\leq \mu\mathds{1}^{A'},\\   \nonumber
    &\hspace{0.7cm} \Tr W^{BC} \leq 0,\\ \nonumber
    &\hspace{0.8cm} Y^{A'BC}, V^{A'BC}\geq 0. 
    \end{align}
\end{proposition}
\begin{proof}
Suppose the optimal solution to the optimization problem in Proposition \ref{non-sdp-relax-for-strong-converse} is achieved at $\{\Lambda^{A'BC}, \rho^{A'}, m \}$, so that $Q^{\textbf{PPT} \,\cap\, \textbf{NS}}_{1}(\mathcal{N},\varepsilon)\leq - \log \min m$. Define $R^{A'BC} = \frac{1}{m}\Lambda^{A'BC}$, and observe that $\{R^{A'BC}, \rho^{A'}, m\}$ constitutes a feasible solution to the
Primal program for determining $\Gamma(\mathcal{N})$. Therefore,
\begin{align*}
    Q_{\Gamma}(\mathcal{N}) &\geq \log \Tr R^{A'BC}(N^{A'BC})^T\\
    & \geq \log \frac{1}{m} \Tr \Lambda^{A'BC}(N^{A'BC})^T \\
    &\geq \log \frac{1}{m} (1-\varepsilon)\\
    & = -\log m  + \log (1-\varepsilon) \\
    & \geq Q^{\textbf{PPT} \,\cap\, \textbf{NS}}_{1}(\mathcal{N},\varepsilon) + \log (1-\varepsilon).
\end{align*}
This concludes the proof of the upper bound. 
The dual SDP is derived using the Lagrange multiplier method, where positive semi-definite operators are associated with inequality constraints, and Hermitian operators are assigned to equality constraints. Specifically, the operators $V^{A'BC},Y^{A'BC}\geq 0$ are assigned to the inequality constraints, while the Hermitian operator $W^{BC}$ and the real multiplier $\mu$ correspond to the equality constraints, as appropriate. We obtain the following Lagrangian:
\begin{align*}
    \Gamma (\mathcal{N}) = &\Tr R^{A'BC}(N^{A'BC})^T \\
    &+ \mu(1-\Tr\rho^{A'}) \\
    & + \Tr{\mathds{1}^{A'}W^{BC}\left(m\abs{A'}^{-1}\mathds{1}^{A'BC} - R^{A'BC}\right)} \\
    & + \Tr{Y^{A'BC}\left( \rho^{A'}\mathds{1}^{BC} +  (R^{A'BC})^{T_{BC}} \right)} \\
   & + \Tr{V^{A'BC}\left( \rho^{A'}\mathds{1}^{BC} -  (R^{A'BC})^{T_{BC}} \right)} \\
     = &\Tr\Big\{R^{A'BC}\big((N^{A'BC})^T - \mathds{1}^{A'}W^{BC} + (Y^{A'BC} -V^{A'BC})^{T_{BC}}\big)\Big\} \\
    &+  \Tr\Big\{\rho^{A'}\big(\Tr_{BC}\big\{(Y^{A'BC} + V^{A'BC})^{T_{BC}} \big\} - \mu\mathds{1}^{A'}\big)\Big\} \\
    & + \mu +  m \Tr W^{BC}.
\end{align*}
The terms multiplied to the primal variables are required to be negative semi-definite and less than or equal to zero, as appropriate. Therefore, the dual SDP involves minimizing $\mu$ subject to
\begin{align*}
     (N^{A'BC})^T - \mathds{1}^{A'}W^{BC} + (Y^{A'BC} -V^{A'BC})^{T_{BC}} &\leq 0, \\
     \Tr_{BC}\big\{(Y^{A'BC} + V^{A'BC})^{T_{BC}} \big\} - \mu\mathds{1}^{A'}&\leq 0, \\
    \Tr W^{BC} &\leq 0.
\end{align*}
Thus, the dual constraints involve the above conditions along with $V^{A'BC},Y^{A'BC}\geq 0$. Note that the Hermicity of $W^{BC}$ is implicitly assumed. Moreover, the primal problem satisfies a strict feasibility condition. For example, one could choose 
\begin{align*}
    R^{A'BC} = \frac{\mathds{1}^{A'BC}}{\abs{A'}\abs{B}\abs{C}}, \quad \rho^{A'} = R^{A'}, \quad m = \frac{1}{\abs{B}\abs{C}}.
\end{align*}
Slater's theorem therefore ensures that strong duality holds (so that the optimal values of the primal and dual problems are equal, and the dual problem has an attainable (finite) optimal solution).
\end{proof}

The following proposition shows that the quantity $Q_{\Gamma}(\mathcal{N})$ is additive: 

\begin{proposition}
\label{additivity}
    For any quantum broadcast channels $\mathcal{N}_1^{A'BC}$ and $\mathcal{N}_2^{A'BC}$, the following holds:
    \begin{align*}
        Q_{\Gamma}(\mathcal{N}_1\otimes \mathcal{N}_2) = Q_{\Gamma}(\mathcal{N}_1) + Q_{\Gamma}(\mathcal{N}_2),
    \end{align*}
    where $Q_{\Gamma}(\mathcal{N})$ is defined in Proposition \ref{strong-1}. 
\end{proposition}
\begin{proof}
    We only need to demonstrate the multiplicativity of the function $\Gamma(\mathcal{N})$, i.e. $\Gamma(\mathcal{N}_1\otimes \mathcal{N}_2) = \Gamma(\mathcal{N}_1)\Gamma(\mathcal{N}_2)$. To establish super-multiplicativity, from the primal program Eq. \eqref{primal}, suppose the optimal solutions for $\Gamma(\mathcal{N}_1)$ and $\Gamma(\mathcal{N}_2)$ are given by $\{R_1,\rho^{A'}_1,m_1\}$ and $\{R_2,\rho^{A'}_2,m_2\}$, respectively. It can then be verified that $\{R_1 \otimes R_2,\rho^{A'}_1\otimes \rho^{A'}_2,m_1m_2\}$ is a feasible solution for the program corresponding to $\Gamma(\mathcal{N}_1\otimes \mathcal{N}_2)$. In particular, we use the fact that for Hermitian operators $A_1,A_2,B_1$ and $B_2$, if $-A_1\leq B_1 \leq A_1$ and $-A_2\leq B_2 \leq A_2$, then it follows that $-A_1 \otimes A_2 \leq B_1 \otimes B_2 \leq A_1 \otimes A_2$. Consequently, we conclude that
    \begin{align*}
        \Gamma(\mathcal{N}_1\otimes \mathcal{N}_2) &\geq \Tr{(R_1 \otimes R_2)(N_1\otimes N_2)}\\
        & = \Gamma(\mathcal{N}_1)\Gamma(\mathcal{N}_2).
    \end{align*}
    To show that $\Gamma(\mathcal{N})$ is sub-multiplicative, we utilize the dual formulation of the problem. Let $\{V_1^{A'BC},Y_1^{A'BC},W_1^{BC},\mu_1\}$ and $\{V_2^{A'BC},Y_2^{A'BC},W_2^{BC},\mu_2\}$ denote the optimal solutions of the dual SDP in Eq. \eqref{dual} corresponding to $\mathcal{N}_1$ and $\mathcal{N}_2$, respectively. Consider the following variables:
    \begin{align*}
        V &= V_1\otimes V_2 + Y_1 \otimes Y_2,\,\, 
        Y = V_1 \otimes Y_2 \otimes Y_1 \otimes V_2, \\
        W &= V_1 \otimes \mathds{1}^{A'} W^{BC}_2 - Y_1 \otimes \mathds{1}^{A'} W^{BC}_2 + \mathds{1}^{A'}W^{BC}_2 \otimes V_2 - \mathds{1}^{A'}W^{BC}_2 \otimes Y_2  + \mathds{1}^{A'}W^{BC}_1 \otimes \mathds{1}^{A'}W^{BC}_2.
    \end{align*}
    It can be verified that these variables constitute a feasible solution for $\Gamma(\mathcal{N}_1\otimes \mathcal{N}_2)$. To establish the feasibility of the last line, we utilize the relation $- \mathds{1}^{A'}W^{BC} \leq  (V^{A'BC} -Y^{A'BC})^{T_{BC}}$. Hence,
    \begin{align*}
        \Gamma(\mathcal{N}_1\otimes \mathcal{N}_2) \leq  \Gamma(\mathcal{N}_1)\Gamma(\mathcal{N}_2).
    \end{align*}
    This completes the proof of the additivity of $\Gamma(\mathcal{N})$. 
\end{proof}

We are now ready to present the main results of this section. We first prove that $Q_{\Gamma}(\mathcal{N})$ is a weak converse bound on the asymptotic sum-capacity of the quantum broadcast channel $\mathcal{N}$, and subsequently we demonstrate that $Q_{\Gamma}(\mathcal{N})$ is, in fact, a strong converse bound.
\begin{theorem}
    For any quantum broadcast channel $\mathcal{N}^{A'\to BC}$, we have
    \begin{align*}
        Q^{\textbf{UN}}(\mathcal{N}) \leq Q^{\textbf{PPT} \,\cap\, \text{NS}}(\mathcal{N}) \leq Q_{\Gamma}(\mathcal{N}),
    \end{align*}
    where $Q_{\Gamma}(\mathcal{N})$ is defined in Proposition \ref{strong-1}.
\end{theorem}
\begin{proof}
    The inequality on the left-hand-side is trivial. For the inequality on the right-hand-side, we regularize both sides of Eq. \eqref{one-shot-converse} as follows:
    \begin{align*}
        Q^{\textbf{PPT} \,\cap\, \textbf{NS}}(\mathcal{N}) 
        &= \lim_{\varepsilon\to 0}\lim_{n\to\infty} Q^{\textbf{PPT} \,\cap\, \textbf{NS}}_{1}(\mathcal{N}^{\otimes n},\varepsilon) \\
         &\leq \lim_{\varepsilon\to 0}\lim_{n\to\infty} \frac{1}{n} \big [Q_{\Gamma}(\mathcal{N}^{\otimes n}) - \log(1-\varepsilon)\big]\\
         & = Q_{\Gamma}(\mathcal{N}),
    \end{align*}
    where the last equality follows from the additivity of $Q_{\Gamma}(\mathcal{N})$ as per Proposition \ref{additivity}.
\end{proof}

\medskip

We now present our strong converse bound, which serves as the key result of this section:

\begin{theorem}
\label{strong-converse}
    For any quantum broadcast channel $\mathcal{N}^{A'\to BC}$, let $R_1$ and $R_2$ represent the asymptotic rates from the sender to receiver $B$ and receiver $C$, respectively. If $R_1 + R_2 \geq Q_{\Gamma}(\mathcal{N})$,
    then the pair $(R_1,R_2)$ forms a strong converse bound for the capacity region of $\mathcal{N}^{A'\to BC}$, with $Q_{\Gamma}(\mathcal{N})$ as defined in Proposition \ref{strong-1}.
\end{theorem}
\begin{proof}
    Assume the rate pair $(R_1,R_2)$ is asymptotically achievable. From Eq. \eqref{one-shot-converse} and the additivity of $Q_{\Gamma}(\mathcal{N})$, we obtain the inequality
    \begin{align*}
        n(R_1+R_2) \leq n Q_{\Gamma}(\mathcal{N}) - \log (1-\varepsilon),
    \end{align*}
which can be rewritten as
    \begin{align*}
        \varepsilon \geq 1- 2^{n\left(Q_{\Gamma}(\mathcal{N}) - R_1-R_2\right)}.
    \end{align*}
    This shows that if $R_1 + R_2 \geq Q_{\Gamma}(\mathcal{N})$, the error $\varepsilon$ will exponentially converge to $1$ as the number of channel uses tends to infinity.
\end{proof}

In \cite[Theorem 7]{wang-q}, a strong converse bound on the quantum capacity of a point-to-point channel is presented. The following proposition demonstrates that our strong converse bound can generally be tighter when applied to point-to-point channels (when the receivers $B$ and $C$ are viewed as a single system $BC$).

\begin{proposition}
    For any point-to-point quantum channel $\mathcal{N}^{A'\to BC}$, where $BC$ is treated as a single receiver, the strong converse bound in Proposition \ref{strong-converse} is generally tighter than its counter part in \cite[Theorem 7]{wang-q}.  
\end{proposition}
\begin{proof}
    Note that the quantity $Q^{\textbf{PPT} \,\cap\, \textbf{NS}}_{1}(\mathcal{N},\varepsilon)$ in proposition \ref{non-sdp-relax-for-strong-converse} can generally be smaller than the corresponding quantity $Q^{\textbf{PPT}}_{1}(\mathcal{N},\varepsilon)$ in \cite[Eq. (23)]{wang-q}. This is due to our inclusion of the no-signalling condition $\Lambda^{BC}=m^2\mathds{1}^{BC}$, yet we were still able to prove the additivity of the quantity $Q_{\Gamma}(\mathcal{N})$.   
\end{proof}

In the next section, we explore how the codes derived so far can aid in entanglement distillation and, conversely, how existing entanglement distillation protocols can be adapted into quantum codes for quantum broadcast channels.

\bigskip

\section{From PPT-preserving codes to entanglement combing and back}
\label{combing}
In this section, we establish a correspondence between PPT-preserving quantum codes and PPT-preserving entanglement combing schemes, our Proposition \ref{code-combing-back}.   
To motivate the problem intuitively, suppose Alice generates an EPR pair of size $\abs{A'}$ and transmits half of it over the quantum broadcast channel $\mathcal{N}^{A'\to BC}$. As a result, the three of them share the state $\widetilde{N}^{A'BC}$, where $A'$, $B$, and $C$ belong to Alice, Bob, and Charlie, respectively. Suppose further that they can apply a particular PPT-preserving entanglement combing protocol to transform $\widetilde{N}^{A'BC}$ into a state $\sigma^{A_1A_2B'C'}$ such that 
\begin{align*}
    \bra{\phi}^{A_1B'}\bra{\phi}^{A_2C'}\sigma^{A_1A_2B'C'}\ket{\phi}^{A_1B'}\ket{\phi}^{A_2C'} = f,
\end{align*}
 shared between Alice $A_1,A_2$, Bob $B$, and Charlie $C$, respectively.  

Can the PPT-preserving entanglement combing protocol serve as a PPT-preserving quantum code for the quantum broadcast channel $\mathcal{N}^{A'\to BC}$? If so, how does the entanglement fidelity of the resulting code compare to that of the original combing protocol?
Conversely, suppose there exists a PPT-preserving code $\mathcal{Z}^{A'B'C'\to ABC}$ of channel fidelity $F(\mathcal{N},r_1,r_2) = f$ for the quantum broadcast channel $\mathcal{N}^{A'\to BC}$. Can this code be converted into a PPT-preserving entanglement combing protocol? If so, how is the entanglement fidelity of the combing protocol related to that of the channel?  In this section, we fully answer these questions.
The importance of this result is two fold — firstly, it allows us to use existing quantum codes for broadcast channels for entanglement combing \cite{8861115, salek2024threereceiverquantumbroadcastchannels}.  Secondly, in cases where a result for a class of PPT-preserving distillation protocols is already known as in \cite{salek2022}, we can potentially translate the result to a related class of PPT-preserving codes without any further work. The second approach is further investigated in \cite{3-salek-hayden}.

In this section, we work with the Choi state rather than the Choi matrix. As noted in the preliminaries, we use a tilde to indicate that the Choi state is derived from a normalized EPR pair, whereas the notation without a tilde refers to the Choi matrix.
We now define the maximum fidelity of a PPT-preserving entanglement combing operation. Recall that these operations correspond to those whose Choi state is PPT across all three cuts; see the introduction for further details.

\begin{definition}
For a tripartite state $\rho^{A'BC}$, its optimal entanglement combing fidelity with PPT-preserving operations is defined as follows:
    \begin{align*}
        F_{\Gamma}^{^{\textbf{PPT}}}(\rho^{A'BC}, r_1,r_2)&\coloneqq \max \Tr  \phi^{A_1B'} \phi^{A_2C'} \mathcal{Y}^{A'BC\to A_1A_2B'C'}(\rho^{A'BC})  \\
        &\hspace{1cm}\text{s.t.}\,\,\,\, \mathcal{Y}^{A'BC\to A_1A_2B'C'} \,\,\text{is PPT-preserving},\\
        &\hspace{1cm}\ket{\phi}^{A_1B'}=\frac{1}{\sqrt{r_1}}\sum_i\ket{i,i}^{A_1B'}, \ket{\phi}^{A_2C'}=\frac{1}{\sqrt{r_2}}\sum_i\ket{i,i}^{A_2C'},\\
        &\hspace{1cm} r_1 = \abs{A_1} = \abs{B'}, r_2 = \abs{A_2} = \abs{C'}.
    \end{align*}
The pair $(r_1,r_2)$ is referred to as the optimal size of the entanglement combing of the state $\rho^{A'BC}$ optimized over all PPT-preserving entanglement combing operations. We define the asymptotic sum-rate of entanglement combing of a state $\rho^{A'BC}$ as follows:

\begin{align*}
    D_{\Gamma}^{^{\textbf{PPT}}}(\rho^{A'BC})\coloneqq \sup\left\{R_1+R_2: \lim_{n\to\infty}F_{\Gamma}^{^{\textbf{PPT}}}\big((\rho^{A'BC})^{\otimes n}, \lfloor{2^{nR_1}}\rfloor,\lfloor{2^{nR_2}}\rfloor\big)=1\right\}.
\end{align*}
\end{definition}

In the following, we generalize the insights from \cite{leung-matthews} to establish a connection between the entanglement combing of the Choi state $\widetilde{N}^{A'BC}$ via tripartite PPT-preserving operations and PPT-preserving error correction codes for the quantum broadcast channel $\mathcal{N}^{A'\to BC}$.

\begin{proposition}
\label{code-combing-back}
    For any quantum broadcast channel $\mathcal{N}^{A'\to BC}$ with Choi state $\widetilde{N}^{A'BC}$,
    \begin{enumerate}[(i)]
        \item If a PPT-preserving entanglement combing operation can distill EPR pairs of size $(r_1,r_2)$ from the tripartite state $\widetilde{N}^{A'BC}$ with entanglement fidelity $f$, where $r_1$ is the dimension of the pair shared between Alice and Bob and $r_2$ that of Alice and Charlie, then there exists a PPT-preserving quantum code of size $(r_1,r_2)$ with channel fidelity $f$ for the broadcast channel $\mathcal{N}^{A'\to BC}$. Consequently, 
        \begin{align*}
            F^{\textbf{PPT}}(\mathcal{N}^{A'\to BC},r_1,r_2) \geq F_{\Gamma}^{\textbf{PPT}}(\widetilde{N}^{A'BC}, r_1,r_2).
        \end{align*}   
        \item If the quantum broadcast channel $\mathcal{N}^{A'\to BC}$ can be exactly simulated using a single copy of its Choi state $\widetilde{N}^{A'BC}$ and forward classical communication from Alice to Bob and Charlie, then the converse to $(i)$ also holds true, therefore
        \begin{align*}
            F^{\textbf{PPT}}(\mathcal{N}^{A'\to BC},r_1,r_2) = F_{\Gamma}^{\textbf{PPT}}(\widetilde{N}^{A'BC}, r_1,r_2).
        \end{align*}
        Moreover,
        \begin{align*}
         Q^{\textbf{PPT}}(\mathcal{N}^{A'\to BC}) =   D_{\Gamma}^{^{\textbf{PPT}}}(\widetilde{N}^{A'BC}),
        \end{align*}
        where $Q^{\textbf{PPT}}(\mathcal{N}^{A'\to BC})$ is defined in Eq. \eqref{asym-sum-capacity}.
    \end{enumerate}
\end{proposition}

\begin{proof}
(i) Suppose there exists a tripartite quantum channel $\mathcal{Y}^{A'BC \to A_1A_2B'C'}$ that is PPT-preserving but otherwise arbitrary, and transforms the state $\widetilde{N}^{A'BC}$ into $\phi^{A_1B'} \otimes \phi^{A_2C'}$, where $\phi^{A_1B'}$ is shared between Alice ($A_1$) and Bob ($B'$), and $\phi^{A_2C'}$ is shared between Alice ($A_2$) and Charlie ($C'$), such that  
\begin{align*}
    \bra{\phi}^{A_1B'}\bra{\phi}^{A_2C'}\mathcal{Y}^{A'BC \to A_1A_2B'C'}\big(\widetilde{N}^{A'BC}\big)\ket{\phi}^{A_1B'}\ket{\phi}^{A_2C'} = f.
\end{align*}  
We recall that the only strict requirement for any tripartite operation to qualify as a physically realizable code is no-signalling from Bob and Charlie to Alice, i.e., $BC \not\to A$. This ensures the causality of the code. Hence, the PPT-preserving combing operation and any subsequent operations used to build a code must satisfy $BC \not\to A$.
However, we can show that any PPT-preserving, otherwise arbitrary, entanglement combing operation can, without loss of generality, be made fully no-signalling by applying a twirling trick similar to the twirling technique we used to derive the twirled Choi matrix in Eq. \eqref{twirled-choi-1}.

\begin{figure}
\centering
\includegraphics[width=0.9\textwidth]{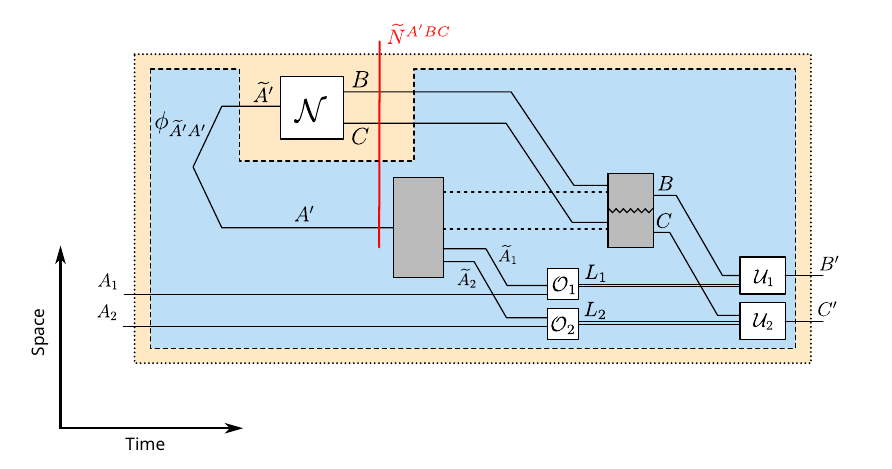}
\caption{Building a PPT-preserving code (the operation inside the dashed blue region) from a PPT-preserving entanglement combing protocol and quantum teleportation. The grey boxes and the dashed lines connecting them represent the PPT-preserving entanglement combing scheme, while the white box operations with double lines for classical messages depicts the teleportation scheme. The entire induced channel is enclosed within the dotted orange region. The protocol operates by first combing the Choi state of the channel and then using it in a teleportation scheme to transmit quantum information. The zigzag line on the side of Bob and Charlie's grey box indicates that they may collaborate, provided their operation remains PPT-preserving.}
\label{comb-to-code}
\end{figure}

 More precisely, we can show that the PPT-preserving operation $\mathcal{Y}^{A'BC \to A_1A_2B'C'}$ can be transformed into a fully no-signalling operation without altering its entanglement combing fidelity. This twirling is achieved using the same unitary group defined by the unitaries $\widehat{U}^{A_1A_2B'C'}$ in Eq. \eqref{concise-U}.
 applying the twirling procedure amounts to averaging $\mathcal{Y}^{A'BC \to A_1A_2B'C'}$ over the conjugate action of $\widehat{U}^{A_1A_2B'C'}$. In detail: 
\begin{align*}
    f &= \bra{\phi}^{A_1B'}\bra{\phi}^{A_2C'}\mathcal{Y}^{A'BC \to A_1A_2B'C'}\big(\widetilde{N}^{A'BC}\big)\ket{\phi}^{A_1B'}\ket{\phi}^{A_2C'} \\
    &= \bra{\phi}^{A_1B'}\bra{\phi}^{A_2C'}(\widehat{U}^{A_1A_2B'C'})^{\dagger}\mathcal{Y}^{A'BC \to A_1A_2B'C'}\big(\widetilde{N}^{A'BC}\big)\widehat{U}^{A_1A_2B'C'}\ket{\phi}^{A_1B'}\ket{\phi}^{A_2C'} \\
    &= \bra{\phi}^{A_1B'}\bra{\phi}^{A_2C'}\left[\int_{U}d\mu(U)
(\widehat{U}^{A_1A_2B'C'})^{\dagger}\mathcal{Y}^{A'BC \to A_1A_2B'C'}\big(\widetilde{N}^{A'BC}\big)\widehat{U}^{A_1A_2B'C'} \right] \ket{\phi}^{A_1B'}\ket{\phi}^{A_2C'},
\end{align*}
where the integral is taken over the Haar measure of the unitary group. This ensures that the resulting operation $\mathcal{Y}$ is fully no-signalling while preserving both its fidelity in entanglement combing and PPT-preserving property.

Applying this twirling, which relies only on shared classical random variables, transforms the PPT-preserving operation into a fully no-signalling operation in all directions. The reason is that the output of the channel, regardless of its form, becomes diagonal in the basis corresponding to the sub-representations of the group. As a result, the marginal state of each party is a maximally mixed state, independent of the initial input state. Consequently, we can assume without loss of generality that PPT-preserving operations are fully no-signalling.

We note that this twirling procedure employs the same technique and group structure as in twirling the code in Eq. \eqref{twirled-code}, but the setting here is different. In the case of the code, there was an input and a channel that connected specific outputs to specific inputs of the tripartite operation. The unitary could be interpreted as Alice applying the unitary to her input, while Bob and Charlie completed their actions after the code was used. In contrast, in the present setting, the three parties apply the unitaries after the tripartite channel has acted.

From this point onward, we aim to construct a code using the tripartite entanglement combing operation, which is PPT-preserving but otherwise allows (potentially quantum) communication from Alice to Bob and Charlie, or between Bob and Charlie, as long as the PPT-preserving condition remains intact. Once again, we can eliminate these signallings via twirling.  

Alice, Bob, and Charlie are given a PPT-preserving entanglement combing operation $\mathcal{Y}^{A'BC \to A_1A_2B'C'}$ with combing fidelity $f$ when applied to the Choi state $\widetilde{N}^{A'BC}$. Their goal is to use this operation to construct a quantum code for the quantum broadcast channel $\mathcal{N}^{A'\to BC}$ with channel fidelity $f$. Referring to Fig. \ref{comb-to-code}, Alice locally prepares an EPR pair $\phi^{A'\widetilde{A}'}$ and transmits half of it through the channel to Bob and Charlie, resulting in the shared state $\widetilde{N}^{A'BC}$. They then apply the operation $\mathcal{Y}^{A'BC \to A_1A_2B'C'}$ to convert this shared state into $\eta^{A_1A_2B'C'}$, which, by assumption, satisfies  
\begin{align*}  
    \bra{\phi}^{A_1B'}\bra{\phi}^{A_2C'}\eta^{A_1A_2B'C'}\ket{\phi}^{A_1B'}\ket{\phi}^{A_2C'} = f.
\end{align*}  

The three parties can now use $\eta^{A_1A_2B'C'}$ in a teleportation scheme to transmit arbitrary states from Alice to Bob and Charlie. It is straightforward to verify that the teleportation channel's fidelity matches the entanglement combing fidelity of the shared state. Since teleportation requires only classical communication from Alice to Bob and Charlie, the causality constraint $BC\not\to A$ remains satisfied.  

The entire process inside the dashed line can be expressed as  
\begin{align*}  
    \mathcal{M}^{A_1A_2 \to BC} (\cdot) = \mathcal{U}_1^{BL_1 \to B'} \mathcal{U}_2^{CL_2 \to C'}\mathcal{O}_1^{A_1\widetilde{A}_1\to L_1}\mathcal{O}_2^{ A_2\widetilde{A}_2\to L_2}\mathcal{Y}^{A'BC \to \widetilde{A}_1\widetilde{A}_2B'C'}\mathcal{N}^{\widetilde{A'}\to BC}\phi^{\widetilde{A'}A'} (\cdot),
\end{align*}  
where $(\cdot)$ represents any state on $A_1A_2$. The sequence above consists entirely of PPT-preserving operations and must itself be PPT-preserving. Given this, and noting that all communication between $A$ and $BC$ is forward-directed, we conclude that the entire operation (outlined by the dashed line in Fig. \ref{comb-to-code} in orange) satisfies the requirements to be a physically realizable PPT-preserving code.

\begin{figure}
\centering
\includegraphics[width=0.8\textwidth]{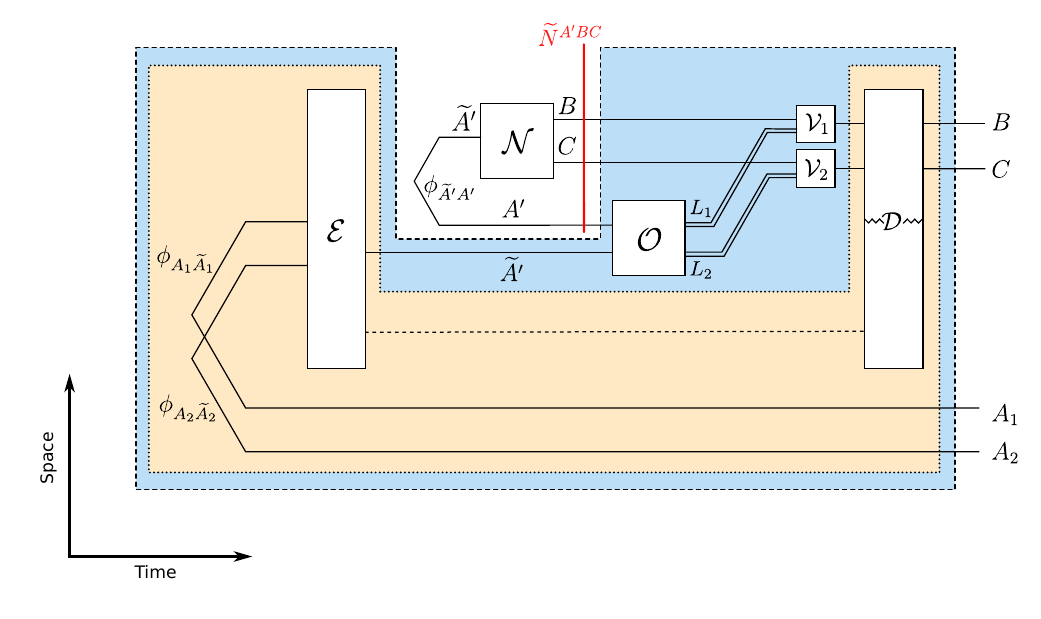}
\caption{Building a PPT-preserving entanglement combing protocol from a PPT-preserving code for a quantum broadcast channel. The code is depicted inside the dotted orange region, while the entanglement combing operation is shown within the dashed blue region. The latter consists of both a coding procedure and a teleportation scheme. The input to the combing region is a tripartite quantum state, which we illustrate as the Choi state of the channel $\mathcal{N}$, whose corresponding code with channel fidelity $f$ is contained within the dotted orange region. The operation involves simulating the action of the channel for the given code. Note that while we use the Choi state $\widetilde{N}^{A'BC}$ as the input to the blue entanglement combing scheme, the input could be any other state. However, different choices of input states lead to different entanglement fidelities. }
\label{fig:code2distillation}
\end{figure}

(ii) Suppose we have a quantum broadcast channel $\mathcal{N}^{A'\to BC}$ that is teleportation simulable. Teleportation simulable channels can be exactly simulated using a single copy of their Choi matrix along with forward classical communication. Referring to Fig. \ref{fig:code2distillation}, this means that
\begin{align*}
\mathcal{N}^{\widetilde{A'}\to BC}(\cdot) = \mathcal{V}{1}^{L_1B\to B}\mathcal{V}{2}^{L_2C\to C} O^{\widetilde{A'}A'\to L_1L_2}\big(\widetilde{N}^{A'BC} \otimes (\cdot)\big),
\end{align*}
where $(\cdot)$ represents an input state on $\widetilde{A'}$ (which is isomorphic to $A'$).
This property holds for certain channels, though a complete characterization remains open—even in the case of point-to-point channels. For further details, the reader is encouraged to consult Sec. \ref{comment-one-way}.

Now, assume there exists a PPT-preserving forward-assisted code $\mathcal{Z}^{ABC\to A'B'C'}$ for this channel, achieving an entanglement fidelity of $F^{\textbf{PPT}}(\mathcal{N}^{A'\to BC},r_1,r_2) = f$.
Rather than treating the channel as a black box, we can replace it with its teleportation simulation. This means that in the code $\mathcal{Z}^{ABC\to A'B'C'}$, the missing step—the action of the channel that transforms $A'$ into $BC$—is explicitly realized through the Choi state. In effect, the distillation of the Choi state occurs by simulating the channel itself. This naturally leads to a PPT-preserving entanglement combing operation with fidelity $f$. The operation remains PPT-preserving because it is constructed from a PPT-preserving code together with forward classical communication (specifically, the messages carried by systems $L_1$ and $L_2$).

Finally, note that if the channel is not teleportation simulable—or if the state we aim to distill is not the Choi state of $\mathcal{N}^{A'\to BC}$ but some other arbitrary tripartite state on $A'BC$—the distillation protocol still applies, but it may fail to achieve the desired entanglement combing fidelity.

\end{proof}

\medskip

\section{Discussion}
\label{discussion}
While the most general quantum codes for creating entanglement between two parties are modeled as bipartite quantum channels, generating entanglement among multiple nodes connected by noisy channels requires an extension of these models. In this work, we study tripartite quantum operations as quantum codes designed to generate EPR pairs among three parties.  
Unlike bipartite operations, tripartite operations have more complex definitions for no-signalling conditions, and depending on the specific no-signalling constraints chosen, an operation may or may not qualify as a valid quantum code for a given task without violating causality. Here, we focus on enforcing a no-signalling condition where two parties—treated as a bundle—cannot signal to the third party. We showed that this approach results in the most general quantum code for creating EPR pairs between a single party and the other two, connected by a single-input, two-output quantum channel. We obtained the channel fidelity as a semidefinite programming (SDP) problem.  

The derived SDP, however, becomes highly nonlinear when the sizes of the EPR pairs become variables themselves, which is the case for quantum capacity. We obtain a hierarchy of semidefinite programming relaxations that result in SDPs. We try to enforce the minimum amount of relaxation possible and explain different ways to obtain weak converse bounds in the form of SDPs. Our strong converse bound, which uses the additivity of a semidefinite program, recovers and strengthens the case of the point-to-point channel.  

Along the way, our results imply new insights and give rise to open problems. For example, since EPR generation is considered secure—ensuring no environment can be correlated with pure states—the bounds here also hold for secret key agreement with two legal receivers and one eavesdropper. The weak and strong converse bounds that we derive for the capacity of this setting will also continue to hold in that case.  

In an upcoming work, we will present lower bounds on the entire rate region of the multiple-receiver quantum broadcast channel using entanglement combing. It would also be interesting to derive new lower bounds on the one-shot capacity of the quantum broadcast channel using one-shot state merging. These lower bounds could then be compared with the upper bounds obtained in this work.  

We are not aware of any lower bounds on the capacity region of the quantum broadcast channel, except for isometric broadcast channels \cite{q-yard}. However, the entanglement combing approach for mixed states \cite{salek2022} can provide lower bounds for quantum broadcast channels. We report these results in \cite{3-salek-hayden}. It is also interesting to compare the combing bounds to the squashed entanglement bounds obtained in \cite{7438836}. This establishes a connection between quantum state merging \cite{state-merging} and quantum broadcast channels. From a slightly different perspective, the relation between state merging and one-way entanglement distillation appears to be vaguely known.  

The operations that satisfy a no-signalling condition, according to the definitions provided by Eqs. \eqref{no-signalling-def-1} or \eqref{no-signalling-def-2}, are sometimes referred to as semicausal in the literature \cite{TEggeling_2002}. The latter reference shows that these no-signalling or semicausal operations are equivalent to classes of operations implementable by local operations and one-way quantum communication in only one direction and refers to them as localizable operations. This terminology was actually introduced earlier in \cite{PhysRevA.64.052309}.  %
Reference \cite[7.1.2]{gour2024resourcesquantumworld}, however, refers to the operations satisfying \eqref{no-signalling-def-1} or \eqref{no-signalling-def-2} as semicausal, while calling the semi-localizable operation no-signalling. 
Causal order and no-signalling correlations have been studied for bipartite channels \cite{Oreshkov2012-xd,PhysRevA.64.052309,TEggeling_2002,PhysRevA.74.012305,6757002}, but a thorough examination of these concepts in tripartite channels is currently lacking. Important questions, such as single-party to single-party no-signalling and its interplay with bipartite no-signalling conditions, remain unexplored. A promising future direction involves investigating the relationship among the trace-preserving condition, no-signalling correlations, and causal order in tripartite channels.  

Another interesting future direction involves applying the approach of \cite{Tomamichel2016} to quantum broadcast channels. This may allow us to determine a lower bound on the number of coherent uses of the broadcast channel necessary to witness super-additivity of the coherent information.

Another frequently asked question in the community is: If one-way distillation is equivalent to an error correction code (see Sec. \ref{comment-one-way}), what is the role of two-way distillation? We comment on this question as follows: two-way distillation cannot constitute a quantum error-correcting code, because a valid code, to be a physically realizable channel, cannot allow messages to be exchanged in both directions—specifically, from Bob to Alice—without violating causality. 
One might ask whether a two-way distillation protocol could be followed by a twirling trick, similar to Proposition \ref{code-combing-back}, to transform the protocol into a fully no-signalling process. However, note that messages from Bob to Alice would have to be exchanged before the use of the channel, which contradicts the required causal order for a valid code. This issue requires further investigation.

\medskip

\appendix
\section{Appendix}
\label{appendix}

\subsection{A Variable-reduced SDP for Theorem \ref{full-fledged} via Fourier-Motzkin Elimination}

The following proposition is equivalent to Theorem \ref{full-fledged}, derived by using Fourier-Motzkin elimination on operators.

\begin{proposition}
    \label{full-fledged-fourier-motzkin}
    For a quantum broadcast channel $\mathcal{N}^{A'\to BC}$, there is a forward-assisted code of size $(r_1,r_2)$, average channel input $\rho^{A'}$ and channel fidelity $f_c$, which is PPT-preserving and/or no-signalling if and only if the following SDP has a feasible solution
    \begin{align*}
     &f_c =
    \Tr E_1^{A'BC}(N^{A'BC})^T\\
     & 0 \leq E^{A'BC}_1\leq \big(E_1^{A' C} + (r_1^2 - 1)E_3^{A' C}\big) \otimes \frac{\mathds{1}^B}{\abs{B}}, \big(E_1^{A' B} + (r_2^2 - 1)E_2^{A' B}\big) \otimes \frac{\mathds{1}^C}{\abs{C}},\\
    & \big(E_1^{A' B} + (r_2^2 - 1)E_2^{A' B}\big) \otimes \frac{\mathds{1}^C}{\abs{C}} + \big(E_1^{A' C} + (r_1^2 - 1)E_3^{A' C}\big) \otimes \frac{\mathds{1}^B}{\abs{B}} \leq  \rho^{A'}\mathds{1}^{BC} + E_1^{A'BC}.\\
     & E_1^{BC}=\frac{\mathds{1}^{BC}}{r_1^2r_2^2},E_3^{C}=\frac{\abs{B}\mathds{1}^{C}}{r_1^2r_2^2},
     E_2^{B}=\frac{\abs{C}\mathds{1}^{B}}{r_1^2r_2^2},
    \end{align*}
    \vspace{-0.3cm}
    \small
    \begin{align*}  
   & \abs{\big(\frac{1}{r_2}(E_1^{A' B})^{T_B} + \frac{(r_2^2 - 1)}{r_2}(E_2^{A' B})^{T_B}\big) \otimes \frac{\mathds{1}^C}{\abs{C}}
 +\frac{1}{r_1}(E_1^{A'C})^{T_C}\otimes\frac{\mathds{1}^{B}}{\abs{B}} + \frac{r_1^2-1}{r_1}(E_3^{A'C})^{T_{C}}\otimes\frac{\mathds{1}^{B}}{\abs{B}}} \leq 
  \frac{1}{r_1r_2}\rho^{A'}\mathds{1}^{BC} + E_1^{T_{BC}},\\
&\abs{\big(\frac{1}{r_2}(E_1^{A' B})^{T_B} + \frac{(r_2^2 - 1)}{r_2}(E_2^{A' B})^{T_B}\big) \otimes \frac{\mathds{1}^C}{\abs{C}}
 - \frac{1}{r_1}(E_1^{A'C})^{T_C}\otimes\frac{\mathds{1}^{B}}{\abs{B}} - \frac{r_1^2-1}{r_1}(E_3^{A'C})^{T_{C}}\otimes\frac{\mathds{1}^{B}}{\abs{B}}} \leq 
  \frac{1}{r_1r_2}\rho^{A'}\mathds{1}^{BC} - E_1^{T_{BC}} ,\\
      &   \abs{E_1^{T_B}} \leq 
        \frac{1}{r_1}\big(E_1^{A' C} + (r_1^2 - 1)E_3^{A' C}\big) \otimes \frac{\mathds{1}^B}{\abs{B}},\\
&\abs{r_1 \big((E_1^{A' B})^{T_B} + (r_2^2 - 1)(E_2^{A' B})^{T_B}\big) \otimes \frac{\mathds{1}^C}{\abs{C}} - r_1 E_1^{T_B}} \leq  \rho^{A'}\mathds{1}^{BC} - \big(E_1^{A' C} + (r_1^2 - 1)E_3^{A' C}\big) \otimes \frac{\mathds{1}^B}{\abs{B}},\\
       &  \abs{E_1^{T_C}} \leq 
        \frac{1}{r_2} \big(E_1^{A' B} + (r_2^2 - 1)E_2^{A' B}\big) \otimes \frac{\mathds{1}^C}{\abs{C}} ,\\
&\abs{r_2 (E_1^{A' C})^{T_C} + r_2(r_1^2 - 1)(E_3^{A' C})^{T_C} \otimes \frac{\mathds{1}^B}{\abs{B}} - r_2 E_1^{T_C}}
\leq  \rho^{A'}\mathds{1}^{BC} - \big(E_1^{A' B} + (r_2^2 - 1)E_2^{A' B}\big) \otimes \frac{\mathds{1}^C}{\abs{C}}.
    \end{align*}
\end{proposition}
\begin{proof}
    This follows from applying Fourier-Motzkin elimination to the constraints of Theorem \ref{full-fledged}. More precisely, we use Eqs. \eqref{BnototAC-one} and \eqref{CnototAB-one} to remove the operators $E_3^{A'BC}$ and $E_2^{A'BC}$, respectively. 
\end{proof} 

\begin{remark}
        Note that the Fourier-Motzkin elimination in the proof of Propostion \ref{full-fledged-fourier-motzkin} cannot proceed further because the next step involves eliminating the operators $E_2^{A'B}$ and $E_3^{A'C}$, which appear alongside their partial transposes. It is generally not feasible to remove $E_2^{A'B}$ and $E_3^{A'C}$ without losing information. 
We explain this via the following simple example. For Hermitian operators $X^{AB},Y^{AB},W^{AB},Z^{AB},$ and $E^{AB}$, assume the following hold:
\begin{align*}
    X^{AB} &\leq E^{AB} \leq Y^{AB},\\
    W^{AB} &\leq (E^{AB})^{T_A} \leq Z^{AB}.
\end{align*}
In the absence of the partial transpose $T_A$, removing the operator $E^{AB}$ via the Fourier-Motzkin elimination results in the following four inequalities $\{X^{AB},W^{AB}\}\leq\{Y^{AB},Z^{AB}\}$. However, removing the partial transpose in general is not possible.
\end{remark}

\medskip

\subsection{On the equivalence of one-way entanglement distillation and quantum error correction}
\label{comment-one-way}

In this section, we clarify a common statement in the literature attributed to \cite{PhysRevA.54.3824}. It is typically mentioned that (possibly asymptotic) distillation of EPR pairs from a bipartite quantum state using one-way classical communication is equivalent to a quantum error correction code for transmitting quantum information. We aim to make this statement precise and identify the conditions under which this equivalence holds.   

Specifically, the so-called equivalence, if valid, relates a quantum code for a given channel $\mathcal{N}^{A'\to B}$ to a one-way distillation protocol for the channel’s Choi state $\widetilde{N}^{A'B}$, and vice versa. However, just as a quantum code optimized for one channel may not perform equally well for another, a distillation protocol derived from a code for the channel $\mathcal{N}^{A'\to B}$ may not achieve the same entanglement fidelity when applied to an arbitrary state rather than $\widetilde{N}^{A'B}$. The reverse also holds: A distillation protocol designed for $\widetilde{N}^{A'B}$ may not be an effective quantum code for a different channel.

Therefore, any one-way distillation protocol must influence the Choi state of the channel $\mathcal{N}^{A'\to B}$. Suppose a one-way entanglement distillation protocol transforms the (normalized) Choi state $\widetilde{N}^{A'B}$ into a state $\sigma^{AB'}$, where $\abs{A}=\abs{B'}$, with entanglement fidelity $f$, i.e., $\bra{\phi}^{AB'}\sigma^{AB'}\ket{\phi}^{AB'} = f$. Then, there exists a quantum code for Alice and Bob with channel fidelity equal to $f$.  
This follows directly: Alice prepares an EPR pair $\phi^{A'\widetilde{A'}}$ and sends $\widetilde{A'}$ through the channel, leading to the shared state $\widetilde{N}^{A'B}$. Applying the one-way distillation protocol, they obtain $\sigma^{AB'}$, which serves as a resource for teleportation. Since the fidelity of a teleportation channel using $\sigma^{AB'}$ as a shared resource matches its entanglement fidelity, Alice and Bob effectively obtain a channel with fidelity $f$, which they can then use to transmit arbitrary quantum states.

On the other hand, and importantly, this does not always hold in the reverse direction. That is, if there exists a quantum code for the channel $\mathcal{N}^{A'\to B}$ such that Alice and Bob can use it to share an EPR pair, resulting in a state $\sigma^{AB'}$ with fidelity $\bra{\phi}^{AB'}\sigma^{AB'}\ket{\phi}^{AB'} = f$, it is not necessarily possible to convert this code into a distillation protocol that extracts $\sigma^{AB'}$ from the Choi state $\widetilde{N}^{A'\to B}$ with the same fidelity $f$. However, this conversion is possible if the channel $\mathcal{N}^{A'\to B}$ is teleportation-simulable, meaning it can be simulated using a single copy of its Choi state.  

Before explaining how a quantum code can be turned into a one-way distillation protocol under teleportation simulability, we first clarify that simulating a channel via a single copy of its Choi state is a probabilistic process.  
Let $N^{A'B}$ be a Choi matrix. If Alice wishes to transmit an arbitrary state $\rho^{A'}$ through the channel, we have  
\begin{align*}
    \Tr_{A'}\left\{N^{A'B} (\rho^{A'})^{T_{A'}}\right\} &= 
    \abs{A'}  \bra{\phi}^{A'\widetilde{A'}}\big(N^{A'B} (\rho^{A'})^{T_{A'}} \otimes \mathds{1}^{\widetilde{A'}}\big )\ket{\phi}^{A'\widetilde{A'}}\\
    &= \abs{A'} \bra{\phi}^{A'\widetilde{A'}}(N^{A'B} \otimes \rho^{A'}) \ket{\phi}^{A'\widetilde{A'}} \\
    &= \abs{A'}^2 \bra{\phi}^{A'\widetilde{A'}}(\widetilde{N}^{A'B} \otimes \rho^{A'}) \ket{\phi}^{A'\widetilde{A'}},
\end{align*}
where the first line follows from Eq. \eqref{trace-EPR}, the second from the transpose trick Eq. \eqref{eq:transposetrick}, and in the last line we use the (normalized) Choi state $\widetilde{N}^{A'B} = \frac{N^{A'B}}{\abs{A'}}$ of the channel $\mathcal{N}^{A'\to B}$.  

This result shows that Alice and Bob can fully simulate the channel’s action using the Choi state as a shared resource, provided Alice's measurement yields the EPR pair $\ket{\phi}=1/\sqrt{\abs{A'}}\sum_i\ket{ii}^{A'\widetilde{A'}}$. However, if Alice obtains any other Bell state, the simulation fails due to the transpose trick used in the second line. More precisely, while the first equality holds for any Bell state, the transpose trick is valid only for $\ket{\phi}$. Furthermore, the factor $\abs{A'}^2$ in the last line corresponds to the probability of Alice obtaining $\ket{\phi}$, which is the same for any Bell state.  

We now explain how a code for a teleportation-simulable channel $\mathcal{N}^{A'\to B}$ leads to a one-way distillation protocol. Alice and Bob have access to a code consisting of an encoder and a decoder, $(\mathcal{E}^{A\to A'},\mathcal{D}^{B'\to B})$, which achieves channel fidelity $f$ when used with $\mathcal{N}^{A'\to B}$. Given that they share the Choi state $\widetilde{N}^{A'B}$, their goal is to apply this code to convert $\widetilde{N}^{A'B}$ into a state with entanglement fidelity $f$. While this protocol could be applied to any bipartite state, we specifically show that, when applied to the Choi state, it achieves the desired fidelity.  

Alice first prepares a maximally entangled Bell state $\ket{\phi}^{A\widetilde{A}}$ and applies her encoder to one half, producing the state $\Theta^{AA'} = \mathcal{E}^{\widetilde{A}\to A'}(\phi^{A\widetilde{A}})$. At this point, the entire state remains in Alice's possession.  
Since Alice and Bob seek to distill the shared state $\widetilde{N}^{A'B}$, they now run a teleportation protocol using this state to transmit the $A'$ system of $\Theta^{AA'}$ to Bob. Because the channel $\mathcal{N}^{A'\to B}$ is teleportation-simulable, this process effectively mimics direct transmission of $A'$ through the channel. Bob then applies his decoder to the received state, realizing the code and producing $\sigma^{AB'}$ with entanglement fidelity $f$. Thus, Alice and Bob can use the pair $(\mathcal{E}^{A\to A'},\mathcal{D}^{B'\to B})$ to distill $\widetilde{N}^{A'B}$ into $\sigma^{AB'}$, leveraging the teleportation simulability of the channel. We note that even if the result of the teleportation is not the Bell state $\phi$, Alice and Bob may still perform post-processing using Alice's classical messages and Bob's unitaries. However, the precise details of this process requires further investigation.

\medskip

\textbf{Acknowledgements}: We are grateful for discussions with Emilio Onorati, Marco Tomamichel, Mark Wilde, Andreas Winter, and Michael Wolf. FS is supported by the Walter Benjamin Fellowship, DFG project No. 524058134. PH was supported by ARO (award W911NF2120214), DOE (Q-NEXT), CIFAR and the Simons Foundation. HR is supported by the Munich Quantum Valley. DL received support from an NSERC discovery grant and also from an NSERC Alliance grant under the project QUORUM (ALLRP-578455-2022). This research was supported in part by Perimeter Institute for Theoretical Physics. Research at Perimeter Institute is supported by the Government of Canada through the Department of Innovation, Science, and Economic Development, and by the Province of Ontario through the Ministry of Colleges and Universities.

\medskip

\bibliographystyle{alpha}
\bibliography{sdp-broadcast}

\end{document}